\documentclass[11pt]{article}
\usepackage{subcaption}
\usepackage{amssymb}
\usepackage{latexsym}
\usepackage{amsmath}
\usepackage{epsfig}
\usepackage{amscd}
\usepackage{amsthm}
\usepackage[normalem]{ulem}
\usepackage{multirow}
\usepackage{float}
\usepackage{footmisc}
\usepackage{color}
\usepackage{enumerate}
\usepackage{bm}
\usepackage{setspace}
\usepackage{adjustbox} 
\usepackage{xparse}
\usepackage{subcaption}
\usepackage{tcolorbox}
\usepackage{lipsum}
\usepackage{blindtext}
\usepackage{mathtools}
\usepackage[titletoc,title]{appendix}
\usepackage[T1]{fontenc}
\usepackage{dsfont}
\usepackage{natbib}
\usepackage{url}
\usepackage{booktabs}
\usepackage{algorithm}
\usepackage{algorithmicx}
\usepackage{algpseudocode}
\usepackage[font=small,labelfont=bf]{caption}
\usepackage{authblk} 
\usepackage{hyperref}[]
\hypersetup{
	colorlinks = true,
	linkcolor = blue!60!black, 
	filecolor = magenta,
	urlcolor = blue!60!black,
	citecolor = blue!60!black 
}
\usepackage{forest}
\usepackage{sectsty}
\usepackage{tikz}
\usepackage{xcolor}

\newcommand{\mbb}[1]{\mathbb{#1}}
\newcommand{\mcl}[1]{\mathcal{#1}}

\newcommand{\opn}[1]{\operatorname{#1}}

\newtheorem{definition}{Definition}
\newtheorem{theorem}{Theorem}
\newtheorem{lemma}{Lemma}

\newtheorem{corollary}{Corollary}

\begin{document}

\begin{center}
	\textbf{\LARGE Multiscale Cochran-Mantel-Haenszel Scanning for Conditional Dependency}

	\vspace*{.2in}

	\begin{author}
		A
		Gyeonghun Kang$^{\dagger}$  \quad
		Jialiang Mao$^{\ddagger}$ \quad
		Li Ma$^{\S}$
	\end{author}

	\vspace*{.2in}

	\begin{tabular}{c}
		$^{\dagger}$Department of Statistical Science, Duke University                    \\
		$^{\ddagger}$Uber Technologies Inc.                                                                 \\
		$^{\S}$Department of Statistics and Data Science Institute, University of Chicago \\
	\end{tabular}

	\vspace*{.2in}

	\today

	\vspace*{.2in}
\end{center}

\begin{abstract}
    We propose a nonparametric approach to testing conditional independence and estimating conditional association, generalizing the Cochran-Mantel-Haenszel (CMH) test and odds-ratio estimator to continuous sample spaces.
    It leverages a multiscale scanning approach to decompose the sample space into a cascade of $2\times 2 \times T$ tables. 
    Following the CMH test, we condition on the marginal order statistics, which are ``almost ancillary'' regarding conditional dependency.
    This strategy helps overcome a key challenge faced by other methods that discretize the sample space: we achieve consistency without requiring stratum sample sizes to grow to infinity, a constraint often difficult to satisfy in practice.
    Our method produces easy-to-compute test statistics with a known asymptotic null distribution under the conditional sampling model, scaling almost linearly with the sample size. 
    Our simulation results demonstrate reliable Type I error control, even with small samples and high-dimensional conditioning, and competitive power compared to state-of-the-art tests.
    Finally, a case study on Uber ride-share data highlights the method’s unique dual capability, inherited from the CMH, to both test and identify the nature of the inferred conditional association.
    By providing summary statistics that capture the strength and direction of local associations, our method offers practitioners a useful tool for learning conditional dependencies.
\end{abstract}

\vskip 2em

\section{Introduction}
Many statistical questions concern the dependence between two sets of random variables, $X$ and $Y$, conditional on another set $Z$.
A core building block in many inference pipelines is the test of conditional independence (CI)---that is, assessing whether there is any remaining association between $X$ and $Y$ given the information contained in $Z$.
CI tests are frequently carried out as a preliminary step to select relevant predictors and reduce the dimensionality of subsequent tasks, such as density estimation and regression modeling. 
They are also employed to identify structural relationships among collections of random variables.
Such applications involve high-dimensional data exhibiting complex dependence patterns, and the number of CI tests can grow at a polynomial rate, if not exponential, in the number of covariates \citep{sondhi2019reduced}.
Against this backdrop, it is advantageous for a nonparametric CI test to be light on assumptions and scalable to large datasets with near-linear time complexity and minimal tuning \citep{azadkia2021simple}.

Recent years have seen steady progress toward this goal.
One such attempt embeds data into reproducing kernel Hilbert spaces (RKHS) to construct tests that are robust to varying functional associations and noise models.
The test statistic is typically an estimator of a quantity that encodes the conditional dependence, such as the norm of the conditional cross-covariance operator \citep{fukumizu2007kernel}, the correlation of residual functions in RKHS \citep{zhang2012kernel, strobl2019approximate}, or the distance between mean embeddings \citep{doran2014permutation, scetbon2022asymptotic}.
However, kernel-based methods are not immediately scalable, typically requiring $O(n^2)$ operations due to matrix inversions \citep{fukumizu2007kernel, zhang2012kernel}.
Low-dimensional approximations have been proposed \citep{strobl2019approximate, scetbon2022asymptotic}, but these methods require much larger sample sizes as the dimensionality of $Z$ increases, as reported in \citet{runge2018conditional, chalupka2018fast}. This trend is also evident in our simulations.
Most critically, kernel- and pairwise distance-based methods are often unsuitable for high-dimensional tests, as their power decays polynomially with dimension and depends heavily on the choice of kernel bandwidth \citep{ramdas2015decreasing}.

CI can alternatively be framed as a regression or prediction problem, allowing one to leverage standard supervised learning algorithms. 
CI can be weakly characterized by the residuals of $X$ and $Y$ being uncorrelated after each is regressed on $Z$ \citep{daudin1980partial}.
Accordingly, test statistics have been constructed from the residuals of two separate regressions \citep{zhang2017feature, zhang2018measuring, shah2020hardness, scheidegger2022weighted}. 
In addition, \citet{burkart2017predictive, chalupka2018fast} proposed using the increase in predictive accuracy of $Y$ when incorporating $X$ alongside $Z$, compared to using $Z$ alone.
If the conditional distribution of $X$ given $Z$ is known or can be accurately estimated, the model-$X$ framework of \citet{candes2018panning} and the conditional permutation test of \citet{berrett2020conditional} provide an elegant way to construct a level-$\alpha$ test based on empirical quantiles of test statistics computed on synthetic null samples.
This scheme has been combined with modern generative models to engender deep learning-based CI tests \citep{bellot2019conditional, yang2025conditional, ren2025score}.
However, nonparametric estimation of the full distribution or the mean function is inherently more challenging than testing a single property, especially in high dimensions \citep[Section~2.10]{ingster2003nonparametric}.
The estimation step can thus become a procedural bottleneck, often requiring substantial tuning, which limits the statistical performance of the resulting test.

Another branch of nonparametric CI tests operates by discretizing $Z$.
This approach stratifies the support of $Z$ into $T$ disjoint strata, $S_1,\cdots, S_T$, based on pairwise proximity \citep{margaritis2005distribution, huang2010testing, canonne2018testing, neykov2021minimax, kim2022local}.
Within each stratum, a local statistic quantifies the association between $X$ and $Y$, and the overall test statistic is formulated as a weighted sum of these values.
These methods implicitly assume that the conditional dependence of $(X,Y)$ given $Z$ is constant within each stratum and that $X\perp Y \mid Z\in S_t$ serves as a valid proxy for $X\perp Y \mid Z$ under sufficiently fine stratification.
In this sense, this category also encompasses other methods relying on nearest neighbors or distance-based clustering to construct their test statistics or simulate null samples, such as local permutation \citep{fukumizu2007kernel, sen2017model, runge2018conditional, huang2022kernel, li2023k}.

However, the existing discretization-based methods face a peculiar paradox: strata must be small enough to control Type I error (T1E) yet large enough to detect local associations \citep{doran2014permutation, strobl2019approximate, berrett2020conditional}.
For a continuous $Z$, discretization averages the conditional distribution within each stratum; consequently, local independence $X\perp Y \mid Z\in S_t$ does not, in general, imply global CI.
This forces strata to constrict in diameter as the sample size increases in order to control the discretization error \citep{kim2022local}.
At the same time, the number of samples within each shrinking stratum must also grow to infinity to ensure adequate power.
While \citet{kim2022local} proved the theoretical existence of such stratum diameters for one- and two-dimensional $Z$, a practical construction remains elusive.
This problem is exacerbated in a high-dimensional setting, where clustering itself is already difficult.

To address this challenge, we invoke the Conditionality Principle:
conditioning on ancillary statistics may eliminate nuisance parameters and reduce the testing space \citep[Chapter~2]{berger1988likelihood}.
For binary $X$ and $Y$, the conditional association is quantified by the odds ratio of a $2\times 2$ table for each value of $Z$, and the marginal probabilities are nuisance parameters.
Under fine stratification, one can design a test statistic that reduces the null hypothesis down to $T$ odds ratios and renders its distribution independent of nuisance parameters by conditioning on marginal totals.
This statistic can also aggregate weak local signals across sparse strata, thereby improving power.
The Cochran-Mantel-Haenszel (CMH) test \citep{cochran1954some, mantel1959statistical} for a $2\times 2\times T$ table provides a classic example; hence, we extend it to continuous $Z$ in arbitrary dimensions.
This allows us to exploit the CMH's \emph{sparse-data asymptotic} \citep[Section~6.4.4]{agresti2013categorical}, where the asymptotic distributions assume an increasing number of strata rather than stratum sizes.
Consequently, our test remains consistent without requiring the per-stratum sample size to diverge---a key advantage over other discretization-based methods.

The same principle extends directly to continuous $X$ and $Y$.
We first apply a recursive dyadic partition to discretize $X$ and $Y$ into $2^{k_1}$ and $2^{k_2}$ bins, respectively.
This discretization, combined with the $T$ strata of $Z$, transforms the data into a three-way $2^{k_1}\times 2^{k_2}\times T$ contingency table.
Remarkably, if the three-way table is conditionally independent, then conditioned on the margin totals, the likelihood of this full table factorizes into a product of likelihoods for $2\times 2\times T$ sub-tables.
Each sub-table serves as a window, varying in location and resolution, possibly overlapping, scanning the $(X,Y)$ support to detect dependency.
Due to factorization, the CMH statistics computed on these windows are mutually independent.
The practical implication is powerful: the global CI test boils down to a multiple testing problem over these independent scanning windows, rendering our method a divide-and-conquer strategy akin to \citet{ma2019fisher, gorsky2022multi}.

This makes our method not just a test but a diagnostic tool: by identifying windows contributing to rejection, it pinpoints the significant regions of the $(X,Y)$ support and provides summary statistics indicating the strength of dependence.
Moreover, examining stratum-specific sample odds ratios could help reveal the complex conditional dependence patterns present in the data.
As we will demonstrate in the data analysis section, these features are particularly valuable in modern large-scale datasets, where the rejection of the null hypothesis is often trivial, but the key insight lies in identifying the patterns underlying that rejection.

Our next contribution is a fast and robust algorithm for stratifying $Z$.
Discretization-based methods, in general, require stratifying $Z$ into strata of similar size.
Existing methods rely on distance-based procedures like $k$-means \citep{kim2022localsupp} or $k$-nearest neighbors \citep{sen2017model, runge2018conditional, huang2022kernel, li2023k}.
However, they are not invariant to monotone transformations, meaning that a simple change of units can alter the test results.
More importantly, the notion of distance is unreliable in high dimensions, as distances tend to concentrate and become nearly indistinguishable \citep{aggarwal2001surprising}.
Instead, we propose a recursive partitioning algorithm based on sample medians.
Starting with the entire space, the algorithm recursively splits each partition at its sample median, iterating through each coordinate of $Z$. 
This simple, rank-based procedure naturally produces hyper-rectangular strata with equal counts.
By avoiding computing distances entirely, it is invariant to unit conversion, robust in high dimensions, and efficient with a computational complexity of only $O(n\log n)$ required for finding medians.

The hardness theorem of \citet{shah2020hardness} establishes that any valid CI test with non-trivial power must operate on a constrained hypothesis space.
Accordingly, we assume that conditional distributions are smooth, a prerequisite for the validity of any discretization-based approach \citep{kim2022local}.
Second, by inheriting the CMH test, our test is tailored to detect alternatives where the conditional association maintains a homogeneous direction across all values of $Z$ within some sub-region of the $(X,Y)$ support.
This assumption is far less restrictive than global homogeneity and is met in many real-world scenarios.
For example, a moderate dosage of a medicine might consistently reduce blood pressure across all patient ages, even if very low or high dosages have inconsistent effects.
This structure also encompasses common unobserved confounders influencing a specific region of $X$ and $Y$.
The payoff for this targeted assumption is a gain in statistical power against such dependencies, as our simulations later confirm.
In essence, our method embeds a plausible assumption on conditional dependence into the hypothesis test, making it better suited to detect local signals that globally-focused tests might miss.

\section{Method}\label{sec:method}

Let $P$ be a joint distribution of a triplet $(X,Y,Z)$ defined on a support $\mcl{X}\times \mcl{Y}\times\mcl{Z}$, and $\mcl{P}$ be the space of all such distributions.
We denote the conditional distribution of $(X,Y)$ given $Z$ as $P_{XY\mid Z}$, and the conditional distributions of $X$ and $Y$ given $Z$ as $P_{X\mid Z}$ and $P_{Y\mid Z}$, respectively.
$P_Z$ is the marginal distribution of $Z$.
Define $\mcl{P}_0\subset \mcl{P}$ as the set of all distributions equipped with absolutely continuous densities $p$ with respect to the Lebesgue measure $\mu$, satisfying $\mcl{P}_0 = \{P \in \mcl{P}:P_{XY\mid Z}(X,Y\mid Z) = P_{X\mid Z}(X\mid Z)P_{Y\mid Z}(Y\mid Z)\}$.
Given independent and identically distributed (iid) samples from $P$, denoted as $\{(X_i, Y_i, Z_i):i\in[n]\}$, we develop a statistical procedure to test the null hypothesis $\mcl{H}_0:P\in\mcl{P}_0$.

\subsection{Binary $X$ and $Y$}
\label{sec:binaryXY}
We first describe our test procedure when $\mcl{X}=\mcl{Y} = \{0, 1\}$ and $\mcl{Z}=\mbb{R}^{d}$, in which case $P_{XY\mid Z}$ is fully characterized by the cell probabilities of a $2\times 2$ table given $Z$.
Accordingly, the data $\{(X_i, Y_i, Z_i):i\in[n]\}$ can be regarded as having been generated from an uncountably infinite mixture of $2\times 2$ tables.
Define the \emph{conditional log odds ratio} at $Z=z$:
\begin{equation*}
	\theta(z) = \log \frac{
		P_{XY\mid Z}(0,0\mid z) P_{XY\mid Z}(1,1\mid z)}{
		P_{XY\mid Z}(0,1\mid z) P_{XY\mid Z}(1,0\mid z)}.
\end{equation*}
Utilizing $\theta(z)$, $\mcl{H}_0$ can be formulated in terms of the log odds ratios of each table as $\mcl{H}_0: \theta(z) \stackrel{a.s.}{=} 0$.
Rather than addressing $\mcl{H}_0$, we stratify $\mcl{Z}$ into $T$ strata, recasting the problem as testing conditional independence within a $2\times 2\times T$ table.

Let $\mcl{S}= \{S_1, S_2, \cdots, S_T\}$ for $T\geq 1$ be a $T$-stratification of $\mcl{Z}$ if it forms a partition of $\mcl{Z}$.
Given $\mcl{S}$, we write $n(x, y, S_t) = |\{i:X_i=x, Y_i=y,Z_i\in S_t\}|$ as the cell counts of the $2\times 2$ table of all $(X_i,Y_i, Z_i)$ in stratum $S_t$, and $n(\cdot, y, S_t)$, $n(x, \cdot, S_t)$, and $n(\cdot, \cdot, S_t)$ as its row, column, and total sums.
In this fashion, we tabulate the data into a $2\times 2\times T$ table.
The conditional distribution of $(X,Y)$ within stratum $S_t$ is a mixture of the pointwise conditional distributions $P_{XY\mid Z}(X,Y \mid Z)$ for $Z\in S_t$, written as $P_{XY\mid Z}(X, Y \mid Z\in S_t) = \mbb{E}_{Z_t}P_{XY\mid Z}(X,Y\mid Z_t)$, where $Z_t$ is a random variable distributed according to $P_Z$ truncated to $S_t$, i.e., $P_Z(dz\mid Z \in S_t) = P_Z(dz)/P_Z(Z_i\in S_t)$, and $\mbb{E}_{Z_t}$ denotes expectation with respect to $Z_t$.
For stratum $S_t$, define the \emph{marginal log odds ratio} within stratum $S_t$:
\begin{equation*}
	\theta_t = \log \frac{
		P_{XY\mid Z}(0, 0 \mid Z\in S_t) P_{XY\mid Z}(1, 1 \mid Z\in S_t)}{
		P_{XY\mid Z}(0, 1 \mid Z\in S_t) P_{XY\mid Z}(1, 0 \mid Z\in S_t)}.
\end{equation*}
The hypothesis $\tilde{\mcl{H}}_0: \theta_1=\cdots=\theta_T=0$ can be readily tested with the CMH test statistic.
Specifically, define $\psi_n=1_{M_n^2 > \chi^2_{\alpha,1}}$ as the test function of a one-sided level-$\alpha$ test of $\tilde{\mcl{H}}_0$ where
\begin{align*}
	M_{n}        & = \frac{\sum_t \left(n(0,0,S_t) - \mu_{t}\right)}{\sqrt{\sum_t \sigma^2_{t}}}, \quad
	\mu_{t} = \frac{n(0,\cdot,S_t) n(\cdot,0,S_t)}{n(\cdot,\cdot,S_t)},                                                               \\
	\sigma^2_{t} & = \frac{n(0,\cdot,S_t) n(1,\cdot,S_t) n(\cdot,0,S_t) n(\cdot,1,S_t)}{n(\cdot,\cdot,S_t)^2 (n(\cdot,\cdot,S_t)-1)}.
\end{align*}
It is well established that if the distribution of the $2\times 2\times T$ table arising from a stratification $\mcl{S}$ belongs to $\tilde{\mcl{H}}_0$, then $M_n^2$ converges in distribution to $\chi_1^2$.
Owing to the single degree of freedom, the convergence is rapid, and---as will be demonstrated in our simulations---the asymptotic approximation is already accurate in relatively small samples.
Indeed, \citet{mantel1980minimum} observed that accuracy is achieved, both under the null and the alternative, once the sum $\sum_t \left(n(0,0,S_t) - \mu_{t}\right)$ can exceed $5$ in magnitude.

The challenge arises because the set of distributions satisfying $\tilde{\mcl{H}}_0$ is not necessarily equivalent to that of $\mcl{H}_0$.
A mixture of independent $2\times 2$ tables can be a dependent table, and averaging over different dependent tables can yield an independent table.
Therefore, it is possible that a conditionally independent distribution $P\in \mcl{P}_0$ with $\theta(z) \stackrel{a.s.}{=}0$ can have a non-zero $\theta_t$ under some stratification $\mcl{S}$.
Conversely, a distribution $Q$ with non-zero $\theta(z)$ for some non-null set of $Z$ can have $\theta_t=0$ for all strata, depending on $\mcl{S}$.
To be precise, consider any $P\in \mcl{P}_0$.
After averaging over each stratum in $\mcl{S}$, the cell counts $n(x, y, S_t)$ are distributed according to $\tilde{P}^n$, where $\tilde{P}$ is a joint distribution of a triplet $(X, Y,\tilde{Z})$, and $\tilde{Z}$ is a discrete variable with a probability mass function (pmf) $s(t) = P_Z\{Z\in S_t\}$:
\begin{equation}\label{eqn:tildep}
	\tilde{P}(X, Y, \tilde{Z}=t) =
	P_Z\{Z \in S_t\} \mbb{E}_{Z_t} \left[P_{X\mid Z} (X\mid Z_t) P_{Y\mid Z} (Y\mid Z_t)\right].
\end{equation}
Note that $\mbb{E}_{Z_t} \left[P_{X\mid Z} (X\mid Z_t) P_{Y\mid Z} (Y\mid Z_t)\right]$ is a mixture of independent $2\times 2$ tables within $S_t$.
This does not, in general, factor into functions of $X$ and $Y$, and it differs from $\tilde{P}_0$ defined as
\begin{equation}\label{eqn:tildep0}
	\tilde{P}_0(X, Y, \tilde{Z}=t) =
	P_Z\{Z \in S_t\} \mbb{E}_{Z_t} \left[P_{X\mid Z} (X\mid Z_t)\right]\mbb{E}_{Z_t}\left[P_{Y\mid Z} (Y\mid Z_t)\right].
\end{equation}
This distribution lies in $\tilde{\mcl{H}}_0$, under which $\psi_n$ converges in distribution to $\chi_1^2$.
Following \citet{kim2022local}, we call $\tilde{P}_0$ the \emph{CI projection} of $\tilde{P}$.
Figure~\ref{fig:CIproj} visualizes the effects of the stratification and the CI projection.
In essence, stratifying $\mcl{Z}$ and applying $\psi_n$ effectively shifts the null away from $\mcl{H}_0$, with the discrepancy measured by the total variation (TV) distance between $\tilde{P}^n$ and $\tilde{P}_0^n$.

\begin{figure}[t!]
	\centering
	\includegraphics[width=1\linewidth]{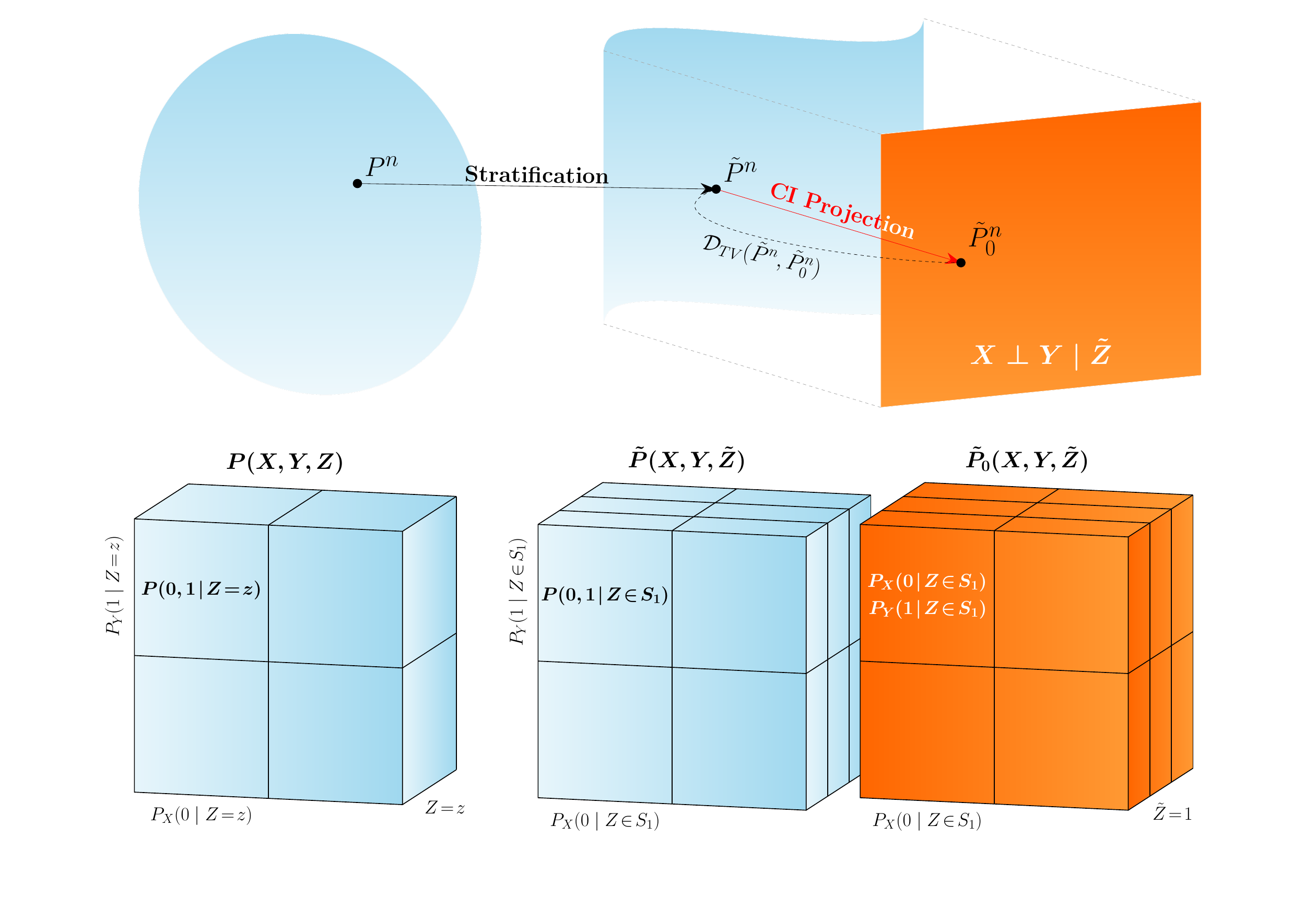}
	\vspace{-1cm}

	\caption{
		This schematic visualizes $P\in \mcl{P}$ along with the induced $\tilde{P}$ and $\tilde{P}_0$.
		$P(X,Y,Z)$ is depicted as a continuum of $2\times 2$ tables for binary $X$ and $Y$.
		$P_X, P_Y$ and $P(0, 1\mid Z=z)$ are shorthands for $P_{X\mid Z}, P_{Y\mid Z}$ and $P_{XY\mid Z}(0, 1\mid Z=z)$.
		A stratification $\mcl{S} = \{S_t:t\in[T]\}$ discretizes $Z$ into $\tilde{Z}\in [T]$ where the event $\{Z\in S_t\}$ is represented as $\{\tilde{Z}=1\}$.
		As such, each cell probability of $\tilde{P}$ is $P(X, Y \mid Z\in S_t)= \mbb{E}_{Z_t} P(X, Y \mid Z_t)$ where $Z_t$ follows $P_Z$ truncated to $S_t$.
		The CI projection factorizes the probability into the product of two marginals, which incurs an error captured by $\mcl{D}_{TV}(\tilde{P}^n, \tilde{P}_0^n)$.}
	\label{fig:CIproj}
\end{figure}

This deviation becomes negligible for an increasingly refined $\mcl{S}$, provided that the conditional distributions are sufficiently smooth.
The TV distance of product distributions can be bounded by the Hellinger distance between $\tilde{P}$ and $\tilde{P}_0$, utilizing inequalities that relate the two metrics.
Specifically, let $\mcl{D}_H(P,Q)$ be the Hellinger distance of two distributions $P$ and $Q$ with densities $p$ and $q$ with respect to $\mu$, written as $\mcl{D}_H(P,Q)=\left(\int |p^{1/2}-q^{1/2}|^2d\mu/2\right)^{1/2}$, and $\delta(z, z')$ be the Euclidean distance between $z,z'\in \mbb{R}^d$.
Under the null $P\in\mcl{P}_0$, the conditional distribution factorizes into two conditionals, and we frame the smoothness condition in terms of these components: $P \in \mcl{P}$ is \emph{marginally smooth} if there exists $L\geq 0$ such that for all $z, z'\in \mcl{Z}$,
$\mcl{D}_H\left(P_{X\mid Z}(X\mid z), P_{X\mid Z}(X\mid z')\right)
\vee
\mcl{D}_H\left(P_{Y\mid Z}(Y\mid z), P_{Y\mid Z}(Y\mid z')\right)\leq L\delta(z,z')$.

\begin{theorem}[Asymptotic T1E control]\label{thm:tbl22 T1E}
	Let $\mcl{S}$ be a $T$-stratification and $h=\max_{S\in\mcl{S}}\sup_{z,z'\in S}\delta(z,z')$.
	If $P \in \mcl{P}_{0}$ is marginally smooth and $h=o(n^{-1/4})$, then for any $\alpha \in [0,1]$, $\lim_{n \to \infty} \mbb{E}_{P^n}[\psi_n] = \alpha$.
\end{theorem}

Under the alternative, our test is consistent for distributions exhibiting homogeneous $\theta(Z)$ across all values of $Z$.
This means the association between $X$ and $Y$ does not change direction across different regions of $\mcl{Z}$.
Note that this homogeneity condition implicitly imposes the smoothness of $Q(X,Y\mid Z)$ where $Q\not\in \mcl{P}_0$.
Indeed, if $Q(X,Y\mid Z)$ is uniformly bounded away from zero for all values of $X$ and $Y$, then the homogeneity condition, together with the marginal smoothness assumption, suffices to attain consistency.
Alternatively, we may require $Q(X,Y\mid Z)$ to be smooth in the following sense: $Q \in \mcl{P}$ is \emph{jointly smooth} if there exists $L\geq 0$ such that for all $z, z'\in \mcl{Z}$, $\mcl{D}_H\left(Q_{XY\mid Z}(X,Y\mid z), Q_{XY\mid Z}(X,Y\mid z')\right) \leq L\delta(z,z')$.
The following theorem establishes consistency under non-negative $\theta(Z)$, and the same conclusion holds when $\theta(Z)$ is non-positive.

\begin{theorem}[Consistency]\label{thm:tbl22 Power}
	Let $\mcl{S}$ and $h$ be defined as in Theorem~\ref{thm:tbl22 T1E}, where $T\asymp n$ and $h=o(n^{-1/2})$.
	Suppose $Q \not\in \mcl{P}_0$ has a conditional distribution $Q(X,Y\mid Z)$ that is continuous in $Z$, and $\theta(Z)$ is non-negative, finite for all $Z\in\mcl{Z}$, and positive with non-zero probability.
	If either (i) $Q$ is jointly smooth or (ii) $Q$ is marginally smooth and $\min_{X,Y} Q(X,Y\mid Z)>q_{min}$ $Q$-a.s. for some $0<q_{min}<1$,
    then $\lim_{n \to \infty} \mbb{E}_{Q^n}[1-\psi_n] = 0$.
\end{theorem}

The proof of Theorem~\ref{thm:tbl22 Power} requires that the total counts in each stratum be bounded in the limit.
This regime is called \emph{sparse-data asymptotics}, in that each stratum remains sparse while the number of strata grows infinitely \citep[Section~6.4.4]{agresti2013categorical}.
The CMH statistic $M_n$ is particularly well suited to this setting: by aggregating deviations from the null across strata, the accumulation of minor differences is enough to reject the null.
In other words, consistency relies on the number of strata rather than the size of a stratum.
This feature distinguishes our test from other discretization methods that require infinitely large counts in each stratum, sacrificing T1E control for consistency.

In accordance with this asymptotic setting, we set $T = \lceil n/\eta \rceil$, where a fixed hyperparameter $\eta \in \mbb{N}$ specifies the desired number of counts in each stratum.
Our simulation results (see Section~\ref{supp-sec:sim3} in the Supplementary Material) empirically verify that the test is consistent for a wide range of $\eta$ under a moderately large number of observations, including small stratum sizes such as $\eta = 5$ and $10$.
We set $\eta=10$ as our default choice.

Computing $M_n$ is straightforward once a stratification is available; the non-trivial step lies in constructing it.
Unlike standard clustering, our aim is not to characterize $P_Z$, but to form spatially connected strata of nearly equal size ($\approx \eta$).
Existing methods such as $k$-means or agglomerative clustering either fail to control stratum size or incur quadratic complexity.
We instead use a procedure based on $k$-$d$ trees \citep{bentley1975multidimensional}, recursively partitioning the sample space by median splits along successive axes.
This yields a stratification into $d$-dimensional hypercubes, with the counts in strata being nearly uniform—specifically, either $\approx n/2^{\lceil \log_2 T\rceil}$ or $\approx 2n/2^{\lceil\log_2 T \rceil}$.
In practice, we set $T = \max(\lceil n/\eta \rceil, 200)$, given that the test statistic converges rapidly.
The resulting algorithm has near-linear complexity, with the dominant cost being the median-finding step of order $O(n \log n)$.
Pseudocode for the procedure, termed \textsc{medtree}, is provided in the Supplementary Material.

\subsection{General $X$ and $Y$}
We propose a divide-and-conquer strategy to test for conditional independence in a general sample space.
Our approach decomposes the global hypothesis into a series of local tests on smaller contingency tables.
The cornerstone of the method is a sequence of nested dyadic partitions of $\mcl{X}$ and $\mcl{Y}$, represented by binary trees.
These partitions allow us to characterize $X\perp Y\mid Z$ as the independence of two-way contingency tables at different values of $Z$, which, in turn, is equivalent to the log odds ratios being zero in all coarse-to-fine $2\times 2$ tables formed along the nodes of these trees.
By discretizing $\mcl{Z}$ into $T$ strata, we further simplify the local hypothesis into testing the zero common log odds ratio in a $2\times 2\times T$ table.
In this sense, we name our method \emph{multiCMH}, as it effectively scans the sample space in varying scales to apply the CMH.
A key insight is that if $(X,Y)$ is conditionally independent given the discretized $Z$, then the $p$-values from these numerous local tests are mutually independent after conditioning on all stratum-specific margin totals.

We first define a dyadic partition and prove how $X\perp Y\mid Z$ can be considered a limiting case of conditional independence on these partitions.
\begin{definition}[Dyadic partition]\label{dfn:binary partition}
	$I_0, I_1, I_2, \cdots $ is a sequence of nested dyadic partitions of $(0,1]$ constructed as follows: starting from $I_0=(0,1]$, let $I_{k+1} = \cup_{I\in I_k} \{I^{left}, I^{right}\}$ where $I=(a, b]$ is a parent node whose left and right child nodes are constructed as $I^{left}=(a, c]$, $I^{right}=(c, b]$ for $c=(a+b)/2$.
\end{definition}
\noindent In our algorithm, we set $c$ as the sample median within $(a,b]$.
We define $\mcl{I}_k = \cup_{i=0}^k I_k$ as a \emph{binary tree} of depth $k$, and $\mcl{I} = \cup_{i=0}^\infty I_k$ as an infinite tree.
It is clear that $\sigma(\mcl{I}_k) \subset \sigma(\mcl{I}_{k+1})$ and $\sigma(\mcl{I}) = \mcl{B}(0,1]$, Borel $\sigma$-algebra on $(0,1]$.
This enables us to formulate an alternative notion of conditional independence based on dyadic partitions:

\begin{definition}[$(k_1,k_2)$-conditional independence]
\label{def:k1k2condind}
	Let $I_0, I_1, \cdots $ and $J_0, J_1, \cdots$ be dyadic partitions of $\mcl{X}$ and $\mcl{Y}$ defined as Definition~\ref{dfn:binary partition}.
	For $(X,Y,Z) \sim P$, we say $X$ and $Y$ are $(k_1,k_2)$-conditional independent given $Z$, written as $X\perp_{k_1,k_2}Y \mid Z$, if $P_{XY\mid Z}(I, J\mid Z) \stackrel{a.s.}{=}P_{X\mid Z}(I\mid Z) P_{Y\mid Z}(J\mid Z)$
	for any $I \in \mcl{I}_{k_1}$ and $J \in \mcl{J}_{k_2}$.
\end{definition}

\noindent It is natural to consider $X \perp Y \mid Z$ as a limiting case of $(k_1,k_2)$-conditional independence.

\begin{theorem}\label{thm:condind}
	$X\perp Y \mid Z$ if and only if
	$X\perp_{k_1,k_2} Y \mid Z$ for any $k_1, k_2 \geq 0$.
\end{theorem}

\noindent Theorem~\ref{thm:condind} implies that, given $\{(X_i, Y_i, Z_i): i \in [n]\}$, one can discretize $X$ and $Y$ based on binary trees of reasonable depths and evaluate the conditional independence of the induced discrete variables.
This approach transforms the original problem into testing for independence across a series of two-way contingency tables---one for each value of $z\in \mcl{Z}$.
By introducing a stratification $\mcl{S}$ on the space $\mcl{Z}$, this collection of two-way tables is consolidated into a single three-way table.

In practice, for deep enough $\mcl{I}_{k_1}$, $\mcl{J}_{k_2}$, and a finer stratification of $\mcl{Z}$, each cell of the induced $2^{k_1} \times 2^{k_2} \times T$ table would contain a small, if not zero, number of counts.
To palliate this, we utilize an alternative characterization of conditional independence in terms of coarse-to-fine $2\times 2 \times T$ tables formed along the parent nodes of $\mcl{I}_{k_1}$ and $\mcl{J}_{k_2}$.
To be specific, we call $I\times J$ a scanning \emph{window} of $\mcl{X}\times \mcl{Y}$, where $I\in \mcl{I}_{k_1-1}$ and $J\in \mcl{J}_{k_2-1}$ are the parent nodes of the binary trees.
Define the \emph{conditional log odds ratio} of a window $I\times J$ at $z\in \mcl{Z}$ as
\begin{equation*}
	\theta(I, J, z) = \log \frac{
		P_{XY\mid Z}(I^{left}, J^{left} \mid z) P_{XY\mid Z}(I^{right}, J^{right} \mid z)
	}{
		P_{XY\mid Z}(I^{left}, J^{right} \mid z) P_{XY\mid Z}(I^{right}, J^{left} \mid z)
	}.
\end{equation*}
The following theorem states that $(k_1,k_2)$-conditional independence holds if and only if the conditional log odds ratios at all such windows are zero.

\begin{theorem}\label{thm:smalltobig}
	$X\perp_{k_1,k_2}Y \mid Z$ if and only if
	$\theta(I, J, Z)\stackrel{a.s.}{=}0$
	for any $I\in \mcl{I}_{k_1-1}$ and $J \in \mcl{J}_{k_2-1}$.
\end{theorem}

Theorem~\ref{thm:smalltobig} inspires our divide-and-conquer strategy: given $\mcl{I}_{k_1}$, $\mcl{J}_{k_2}$, and $\mcl{S}$, we tabulate $\{(X_i, Y_i, Z_i): i \in [n]\}$ into coarse-to-fine $2\times 2\times T$ tables along the parent nodes of the partition trees, for which we conduct the CMH test and compute the corresponding $p$-values.
This prompts the question of the joint distribution of the $p$-values.
Although one might expect the $p$-values from nested or overlapping windows to be correlated, we establish the contrary.
If the overall three-way table is conditionally independent, then, conditioned on row and column sums at each stratum, all $p$-values are mutually independent.

\begin{table}[t!]
	\centering
	\begin{tabular}{lll}
		\toprule
		\textbf{Notation}                 & \textbf{Definition}                                                    & \textbf{Remark}                                   \\
		\midrule
		$n(I,J,S_t)$                      & $|\{i : X_i \in I,\, Y_i \in J,\, Z_i \in S_t\}|$                      & count in a cell $I\times J \times S_t$            \\
		\midrule
		$n(I_{k_1},J_{k_2}, S_t)$         & $\{n(I, J, S_t): I \in \mathcal{I}_{k_1},\, J \in \mathcal{J}_{k_2}\}$ & two-way table $2^{k_1} \times 2^{k_2}$            \\
		$n(I_{k_1},J_{0}, S_t)$           &                                                                        & column sums at stratum $S_t$                      \\
		$n(I_{0},J_{k_2}, S_t)$           &                                                                        & row sums at stratum $S_t$                         \\\midrule
		$n(I_{k_1},J_{k_2}, \mathcal{S})$ & $\{n(I_{k_1},J_{k_2}, S_t): S_t \in \mathcal{S}\}$                     & three-way table $2^{k_1} \times 2^{k_2} \times T$ \\
		$n(I_{k_1},J_{0}, \mathcal{S})$   &                                                                        & column sums at all strata                         \\
		$n(I_{0},J_{k_2}, \mathcal{S})$   &                                                                        & row sums at all strata                            \\
		\bottomrule
	\end{tabular}
	\caption{Notations for contingency tables tabulated from $\{(X_i, Y_i, Z_i):i\in[n]\}$.}
	\label{tab:notations}
\end{table}

Table~\ref{tab:notations} introduces several necessary notations.
Suppose the data consist of iid samples of $P \in \mcl{P}_0$.
A $T$-stratification $\mcl{S}$ discretizes $Z$ into a discrete variable $\tilde{Z} \in [T]$ with a pmf $s(t) = P_Z\{Z \in S_t\}$.
Hence, $(X, Y, \tilde{Z})$ follows $\tilde{P}$, defined as \eqref{eqn:tildep}.
As discussed in Section~\ref{sec:binaryXY}, $X\perp_{k_1,k_2} Y \mid Z$ does not necessarily imply $X \perp_{k_1,k_2} Y \mid \tilde{Z}$.
Instead, we consider $\tilde{P}_0$ defined as \eqref{eqn:tildep0}, under which the probability of a cell $I\times J \times S_t$ becomes $P_Z\{Z \in S_t\} \mbb{E}_{Z_t}\left[ P_{X\mid Z}(I\mid Z_t) \right] \mbb{E}_{Z_t}\left[ P_{Y\mid Z}(J\mid Z_t) \right]$.
Therefore, if we assume $(X_i, Y_i, \tilde{Z}_i)$ is distributed according to $\tilde{P}_0$, then the resulting two-way table at each stratum is independent.
Similar to Theorem~\ref{thm:tbl22 T1E}, the error incurred by substituting $\tilde{P}_0$ for $\tilde{P}$ can always be controlled with finer stratification for distributions with smooth conditional densities.

Another key consequence of employing $\tilde{P}_0$ is that the conditional distribution of $n(I_{k_1}, J_{k_2}, \mcl{S})$ given a collection of stratum-specific margins $n(I_{k_1}, J_{0}, \mcl{S})$ and $n(I_{0}, J_{k_2}, \mcl{S})$ decomposes into a product of $T$ independent distributions---one for each stratum $S_t$.
Each component is a conditional distribution of the two-way table $n(I_{k_1}, J_{k_2}, S_t)$ given the $S_t$-specific margin sums $n(I_{k_1}, J_{0}, S_t)$ and $n(I_{0}, J_{k_2}, S_t)$, which is a multivariate Fisher's hypergeometric (MHG) distribution.
As shown in \citet{ma2019fisher}, the MHG of each stratum can be factorized into a product of simpler Fisher's hypergeometric (HG) distributions of the coarse-to-fine $2\times 2$ tables.

\begin{theorem}[Multiscale factorization]\label{thm:IJKfacto}
	Let $\mbb{P}_{\tilde{P}_0^n}\left\{ n(I_{k_1}, J_{k_2}, \mcl{S}) \mid n(I_{k_1}, J_{0}, \mcl{S}), n(I_{0}, J_{k_2}, \mcl{S}) \right\}$ be the probability of observing a three-way table $n(I_{k_1}, J_{k_2}, \mcl{S})$ conditioned on the collection of margin totals when $\{(X_i, Y_i, \tilde{Z}_i):i\in [n]\}$ are iid samples of $\tilde{P}_0$. Then
	\begin{align}\label{eqn:dag}
		\begin{split}
			 & \mbb{P}_{\tilde{P}_0^n}\left\{ n(I_{k_1}, J_{k_2}, \mcl{S})
			\mid n(I_{k_1}, J_{0}, \mcl{S}), n(I_{0}, J_{k_2}, \mcl{S}) \right\} \\
			 & =
			\prod_{\substack{I \in \mcl{I}_{k_1-1}, J \in \mcl{J}_{k_2-1}}}
			\left\{
			\prod_{S\in \mcl{S}}
			g_0\left(
			n(I^{left},J^{left}, S) \mid
			n(I^{left}, J, S), n(I, J^{left}, S), n(I,J,S) \right)
			\right\}.
		\end{split}
	\end{align}
    where $g_0(a\mid b, c, d)$ is the pmf of HG.
\end{theorem}

\begin{figure}[t!]
	\centering
	\includegraphics[width=0.8\linewidth]{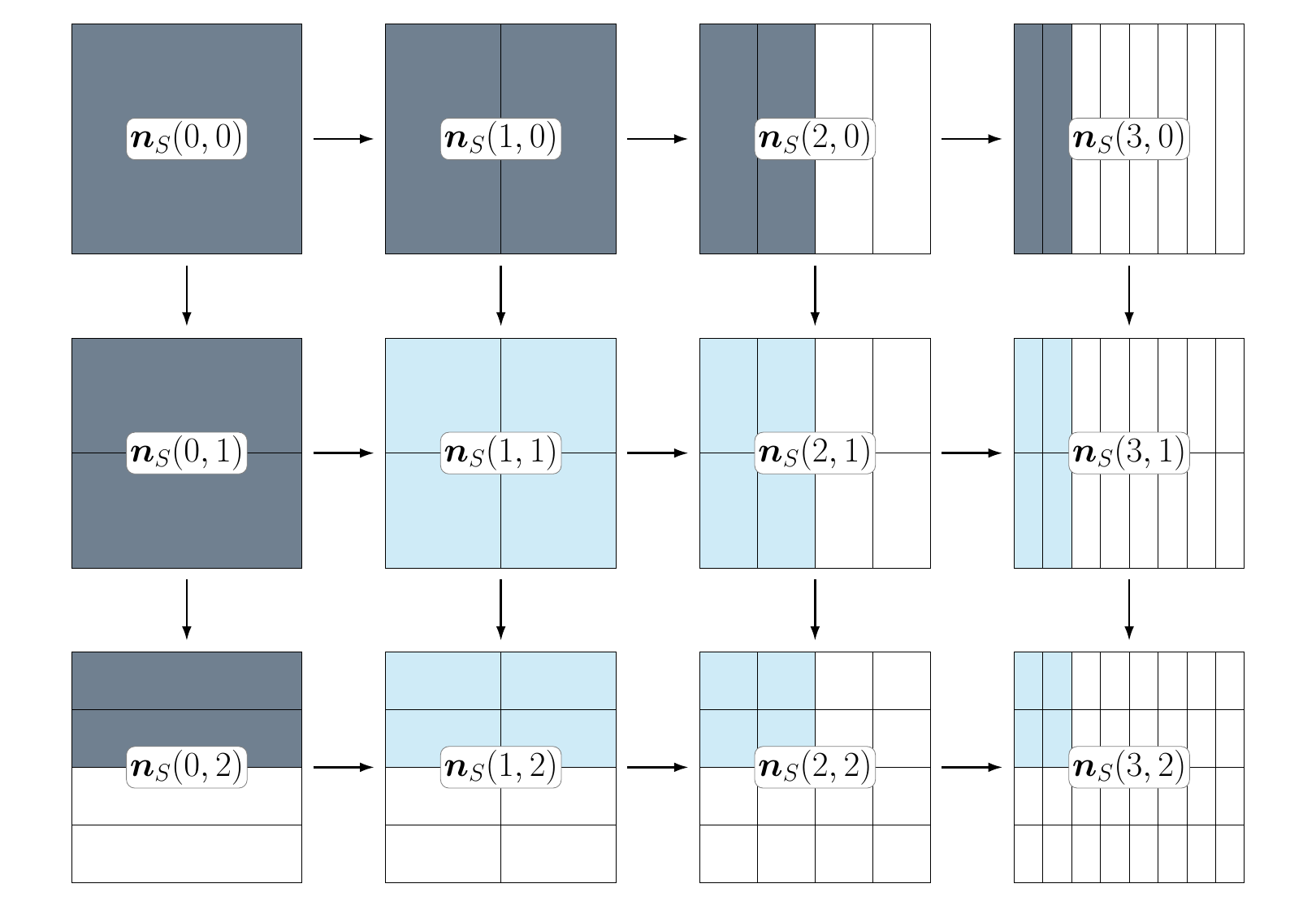}
	\caption{
		A pictorial illustration of the DAG for a three-way table $n(I_{3}, J_2, \mcl{S})$ in Theorem~\ref{thm:IJKfacto}.
		$\bm{n}_S(i,j)$ is a shorthand for a two-way table $n(I_i, J_j, S)$ at stratum $S\in \mcl{S}$.
		$\bm{n}_S(i-1,j)$ constitutes row sums, $\bm{n}_S(i,j-1)$ column sums of all the $2\times 2$ tables of $\bm{n}_S(i,j)$.
		For example, given the margin totals $\bm{n}_S(3,0)$ and $\bm{n}_S(0,2)$, the skyblue colored cell in $\bm{n}_S(1,1)$ is sampled conditioned on the cells marked in gray in $\bm{n}_S(0,1)$ and $\bm{n}_S(1,0)$.
		The process proceeds along the edges of the DAG to generate $\bm{n}_S(3,2)$.
		Repeating this process for all strata yields a draw of $n(I_{3}, J_2, \mcl{S})$ conditioned on $n(I_3, J_0, \mcl{S})$ and $n(I_0, J_2, \mcl{S})$.
	}
	\label{fig:dag}
\end{figure}

Under $\tilde{P}_0$ and conditioned on all stratum-specific margins, the joint distribution of the $2\times 2\times T$ tables formed at each node of the binary trees follows a directed acyclic graph (DAG) structure specified by \eqref{eqn:dag}; see Figure~\ref{fig:dag} for an example when $k_1=3$ and $k_2=2$.
To illustrate, consider a three-way table $n(I_{l_1}, J_{l_2}, \mcl{S})$ and let $l_1+l_2$ be its \emph{resolution}.
$n(I_{l_1}, J_{l_2}, \mcl{S})$ arises from a collection of non-overlapping windows $I\times J$ ($I\in I_{l_1-1}$, $J\in J_{l_2-1}$) scanning $\mcl{X}\times \mcl{Y}$, each corresponding to a $2\times 2\times T$ table $n(\{I^{left}, I\}, \{J^{left}, J\}, \mcl{S})$.
The parents of $n(I_{l_1}, J_{l_2}, \mcl{S})$ in the DAG, which are $n(I_{l_1-1}, J_{l_2}, \mcl{S})$ and $n(I_{l_1}, J_{l_2-1}, \mcl{S})$, constitute the collections of the row and column sums of such $2\times 2\times T$ tables.
The conditional distribution of the child $n(I_{l_1}, J_{l_2}, \mcl{S})$ given the parents is a product of independent $2\times 2\times T$ tables.

As a direct implication of Theorem~\ref{thm:IJKfacto}, if the test statistic of the local hypothesis of each window is a function of the corresponding $2\times 2\times T$ table only, then the $p$-values are mutually independent, conditioned on stratum-specific margin totals.
Since each window is conditionally independent under $\tilde{P}_0$, the corresponding CMH statistic converges in distribution to $\chi_1^2$.
All combined, the $p$-values of CMH tests on all coarse-to-fine $2\times 2\times T$ tables are asymptotically independent.
The following theorem summarizes this result.

\begin{corollary}[Independence of $p$-values]\label{cor:pvalindep}
	Suppose $p(I, J)$ is a $p$-value of any window $I\times J$ computed as a function of $2\times 2\times T$ table $n(\{I^{left}, I\}, \{J^{left}, J\}, \mcl{S})$ such that
	\begin{equation*}
		\lim_{n\to\infty}\mbb{P}_{\tilde{P}_0^n}\left\{
		p(I, J) \leq \alpha(I, J) \mid n(I^{left}, J, \mcl{S}), n(I, J^{left}, \mcl{S}), n(I, J, \mcl{S})
		\right\} = \alpha(I, J),
	\end{equation*}
	where $\alpha(I, J)$ is a significance level of a window $I\times J$.
	Then the following holds:
	\begin{align*}
		 & \mbb{P}_{\tilde{P}_0^n}\bigg\{
		\bigcap_{\substack{I \in \mcl{I}_{k_1-1}, J \in \mcl{J}_{k_2-1}}}
		\{p(I, J) \leq \alpha(I, J)\}\mid
		n(I_{k_1}, J_{0}, \mcl{S}), n(I_{0}, J_{k_2}, \mcl{S})
		\bigg\}                           \\& =
		\prod_{\substack{I \in \mcl{I}_{k_1-1}, J \in \mcl{J}_{k_2-1}}} \mbb{P}_{\tilde{P}_0^n} \left[
			p(I, J) \leq \alpha(I, J) \mid
			n(I^{left}, J), n(I, J^{left}), n(I, J)
			\right]                           \to
		\prod_{\substack{I \in \mcl{I}_{k_1-1}, J \in \mcl{J}_{k_2-1}}} \alpha(I, J).
	\end{align*}
\end{corollary}

\subsection{Multiscale CMH scanning}\label{sec:cmh scanning}

\begin{algorithm}[t!]
	\caption{Multiscale CMH (multiCMH) with three-stage Šidák correction}
	\label{alg:main}
	\begin{algorithmic}[1]
		\State $k_1', k_2' \leftarrow k_1-1, k_2-1$

		\For{$k = 0, 1, 2, \dots, k_1'+k_2'$} \Comment{Scan from low to high resolutions}
		\State $U(k) \leftarrow 0$ 
		\For{$l_1 = \max(0, k-k_2'), \cdots, \min(k_1', k)$} \Comment{Scan each partition}
		\State $l_2 \leftarrow k - l_1$
		\State $L(l_1,l_2) \leftarrow 0$ 

		\For{each window $I\times J$ ($I \in I_{l_1}, J \in J_{l_2}$)} \Comment{Scan each window}
		\If{$V(I, J) = 1$}
		\State $\mcl{S}_{IJ} \leftarrow \Call{medtree}{\{z_i:x_i\in I, y_i\in J\},\eta}$
		\State Compute $p(I, J, \mcl{S}_{IJ})$
		\State $L(l_1,l_2) \leftarrow L(l_1,l_2) + 1$
		\EndIf
		\EndFor

		\If{$L(l_1,l_2) > 0$}
		\State Compute $\tilde{p}_{l_1, l_2}$
		\Comment{Partition multiplicity control}
		\State $U(k) \leftarrow U(k) + 1$
		\EndIf
		\EndFor

		\If{$U(k) > 0$}
		\State Compute $\tilde{p}_{k}$
		\Comment{Resolution multiplicity control}
		\EndIf
		\EndFor

		\State Compute $\tilde{p}$ and reject the null at level $\alpha$ if $\tilde{p} \leq \alpha$
		\Comment{Overall multiplicity control}
		\State Report significant windows where $p(I, J, \mcl{S}_{IJ}) \leq \alpha_n(I, J)$
	\end{algorithmic}
\end{algorithm}

We are now ready to describe multiCMH.
Let $F_{\chi^2_1}$ be the cdf of the $\chi^2_1$ distribution.
For a window $I\times J$ and a $T_{IJ}$-stratification $\mcl{S}_{IJ}$ specific to a window $I\times J$, the $p$-value of the one-sided CMH test is given as $p(I, J, \mcl{S}_{IJ})=1 - F_{\chi^2_1}\left(M_n^2(I, J, \mcl{S}_{IJ})\right)$, where
\begin{align*}
	M_{n}(I, J, \mcl{S}_{IJ}) & =
	\frac{\sum_{S\in \mcl{S}_{IJ}} \left(n(I^{left},J^{left},S) - \mu_{S}\right)}{\sqrt{\sum_{S\in \mcl{S}_{IJ}} \sigma^2_{S}}}, \quad
	\mu_{S} = \frac{n(I^{left},J,S) n(I,J^{left},S)}{n(I,J,S)},                                                                      \\
	\sigma^2_{S}              & = \frac{n(I^{left},J,S) n(I^{right},J,S) n(I,J^{left},S) n(I,J^{right},S)}{n(I,J,S)^2 (n(I,J,S)-1)}.
\end{align*}
Given a finite $k_1$ and $k_2$, we let a $T$-stratification $\mcl{S}$ be the finest common refinement of all the window-specific stratifications $\{\mcl{S}_{IJ}:I\in \mcl{I}_{k_1-1}, J\in \mcl{J}_{k_2-1}\}$.
Since $\mcl{S} \preceq \mcl{S}_{IJ}$, $M_{n}(I, J, \mcl{S}_{IJ})$ is also a function of $n(\{I^{left}, I\}, \{J^{left}, J\}, \mcl{S})$; hence, by Corollary~\ref{cor:pvalindep}, the $p$-values $p(I, J, \mcl{S}_{IJ})$ are asymptotically mutually independent across all windows conditioned on stratum-specific margin totals under $\tilde{P}_0$.

Our inference recipe consists of three steps: screening, $p$-value computation, and multiplicity adjustment.
See Algorithm~\ref{alg:main} for an overview of our algorithm.
Let $n_{IJ}=|\{i:X_i\in I, Y_i\in J\}|$ be the number of observations in a window $I\times J$.
$n_{IJ}$ will be smaller for higher resolutions, with many cells potentially empty within the corresponding $2\times 2\times T$ table.
The $p$-values of such windows are likely to be far from converging to their asymptotic distributions for a limited sample size, incurring only extra penalties in multiplicity correction.
Therefore, we sift out such redundant windows with a screening rule $V(I,J)\in\{0,1\}$:
specifically, $V(I,J)=0$ if either the total count $n_{IJ}$ is less than $v_{all}$ or if any of the margins aggregated over all strata fall below $v_{margin}$.
If $V(I,J)=1$, we proceed to set $T_{IJ} = \lceil n_{IJ}/\eta\rceil$, where $\eta \in \mbb{N}$ is the desired number of observations in each stratum, construct $\mcl{S}_{IJ}$, and compute $p(I, J, \mcl{S}_{IJ})$.
By default, we set $v_{all}=20$ and $v_{margin}=10$.

Given the independence of the $p$-values, any multiplicity correction procedure can be employed to control the family-wise error rate (FWER), which is equivalent to T1E of the global null hypothesis.
However, the hierarchical structure of the dyadic trees should be taken into account.
Simultaneously applying a simple Bonferroni correction across all windows would impose an undue penalty on larger windows.
For this reason, our algorithm applies Šidák's correction \citep{sidak1967rectangular} in three hierarchical stages, as in \citet{ma2019fisher}: first, to windows within the same product partition $I_{l_1} \times J_{l_2}$ of $\mcl{X}\times \mcl{Y}$; second, to the corrected $p$-values across all partitions of the same resolution; and lastly, across all resolutions.
Specifically, for each product partition $I_{l_1} \times J_{l_2}$, we define the \emph{partition-wise} $p$-value as $\tilde{p}_{l_1, l_2} = 1 - \left(1 - \min_{I\in I_{l_1}, J \in J_{l_2}, V(I,J)=1} p(I, J, \mcl{S}_{IJ})\right)^{L(l_1,l_2)}$,
where $L(l_1,l_2)$ is the number of windows with valid ($V(I,J)=1$) $p$-values.
After computing $\tilde{p}_{l_1, l_2}$ for all partitions $(l_1, l_2)$ of the same resolution $l_1+l_2=k$, we compute the \emph{resolution-wise} $p$-value as $\tilde{p}_{k} = 1 - \left(1 - \min_{l_1+l_2=k, L(l_1, l_2)>0} \tilde{p}_{l_1, l_2} \right)^{U(k)}$,
where $U(k)$ is the number of valid ($L(l_1,l_2)>0$) partition-wise $p$-values.
Lastly, using the $p$-values of all resolutions from zero up to $k_1+k_2-2$, we compute the \emph{overall} $p$-value as $\tilde{p} = 1 - \left(1 - \min_{0\leq k \leq k_1+k_2-2} \tilde{p}_{k} \right)^{k_1+k_2-1}$.
This hierarchical procedure controls the FWER by construction and ensures that $p$-values are only corrected against others of the same granularity.

\begin{theorem}[Asymptotic level control]\label{thm:mainT1E}
	Suppose there exists $\mcl{S}$, a common $T$-stratification of $\mcl{Z}$ such that $\mcl{S} \preceq \mcl{S}_{IJ}$ and $T \asymp T_{IJ}$ for any $I\in \mcl{I}_{k_1-1}$ and $J \in \mcl{J}_{k_2-1}$.
	Let $h=\max_{S \in \mcl{S}} \sup_{z,z'\in S}\delta(z,z')$ be its maximal stratum diameter.
	If $P \in \mcl{P}_{0}$ is marginally smooth and $h = o(n^{-1/4})$, then for any $\alpha\in(0,1)$, $\lim_{n\to \infty} \mbb{P}_{P^n}\{\tilde{p} \leq \alpha \} = \alpha$.
\end{theorem}

Beyond simply rejecting the global null hypothesis, the compositional nature of our test allows us to trace a rejection of the null back to the windows that contributed to it.
Šidák's correction can be re-expressed as an adjustment to the significance level for each individual window: given $\alpha$, the multiplicity-corrected significance level of a window $I\times J$ in partition $I_{l_1} \times J_{l_2}$ of resolution $l_1+l_2=k$ is $\alpha_n(I, J) = 1 - (1 - \alpha)^{1 / [(k_1+k_2-1) \cdot U(k) \cdot L(l_1,l_2)]}$.
We define windows where $p(I,J, \mcl{S}_{IJ}) \leq \alpha_n(I, J)$ as \emph{significant windows}.
Taken together, these windows form a map of the conditional dependency structure.
Indeed, if the conditional log odds ratio within a window is either simultaneously positive or negative with a non-zero probability, then the power of our test to reject the local null hypothesis converges to one as the sample size increases.
This local consistency leads immediately to global consistency, even when the resolution is allowed to grow with the sample size.

\begin{theorem}[Local and global consistency]\label{thm:mainPower}
	For $Q \in \mcl{P}$, consider its truncation to $I\times J \times \mcl{Z}$ for $I\in \mcl{I}_{k_1-1}$ and $J \in \mcl{J}_{k_2-1}$.
	If the induced distribution of a triplet $(1_{X\in I^{left}}, 1_{Y\in J^{left}}, Z)$ satisfies the conditions of Theorem~\ref{thm:tbl22 Power}, and if $T\asymp n$ and $h=o(n^{-1/2})$, then for any $\alpha\in(0,1)$,
	$\lim_{n\to\infty} \mbb{P}_{Q^n}\{ p(I, J, \mcl{S}_{IJ}) \leq \alpha_n(I, J) \mid n(I_{k_1}, J_{0}, \mcl{S}_{IJ}), n(I_{0}, J_{k_2}, \mcl{S}_{IJ}) \} = 1$.
	If one or more such windows exist, and if $k_1, k_2$ are either fixed or of order $\mcl{O}(\log n)$, then $\lim_{n\to \infty} \mbb{P}_{Q^n} \{\tilde{p} \leq \alpha \} = 1$ for any $\alpha\in(0,1)$.
\end{theorem}

We conclude our recipe with a note on selecting the partition depth $k_1$ and $k_2$.
Although successive median splits allow for a maximum depth of $\lceil \log_2 n \rceil$, such granularity is unnecessary in practice.
Since our screening rule skips windows with minimum marginal counts (aggregated across strata) below $v_{margin}$, a more practical choice is $\lceil \log_2 (n/v_{margin}) \rceil$.
Alternatively, one may impose a manual cap by setting $\max(k_{max}, \lceil \log_2 (n/v_{margin}) \rceil)$, where $k_{max}$ is chosen so that partitions of size $2^{-k_{max}}$ are sufficiently fine-grained for finite samples.
By default, we take $k_{max}=7$, corresponding to a partition width of about $2^{-7}\approx 0.008$ on the empirical scale.

\section{Simulation Studies}
We conduct two simulations to assess the performance of multiCMH relative to other state-of-the-art methods.
First, we evaluate the finite-sample performance of multiCMH in terms of T1E control and statistical power compared to other methods.
Second, we demonstrate its computational efficiency and scalability on datasets with sample sizes exceeding one million.
All simulations were executed on a high-performance computing cluster using CPU nodes equipped with Intel Xeon Gold 6226 processors.
The code to reproduce the figures is at \url{https://github.com/hun-learning94/multiCMH}.

For reproducible comparisons that are relevant for practitioners, the competing methods were selected based on their public availability and active maintenance.
These methods include the conditional distance correlation (CDIT) \citep{wang2015conditional}, a classifier-based test (CCIT) \citep{sen2017model}, a conditional mutual information statistic estimated based on $k$-nearest neighbors (CMIknn) \citep{runge2018conditional}, the randomized conditional independence test (RCIT) \citep{strobl2019approximate}, the generalized covariance measure (GCM) \citep{shah2020hardness} and its weighted version (wGCM) \citep{scheidegger2022weighted}, and a kernel mean embedding distance (LPCIT) \citep{scetbon2022asymptotic}.

\subsection{Simulation 1: T1E and power analysis}
\label{sec:sim1}

\begin{figure}[t!]
	\centering
	\includegraphics[width=0.95\linewidth]{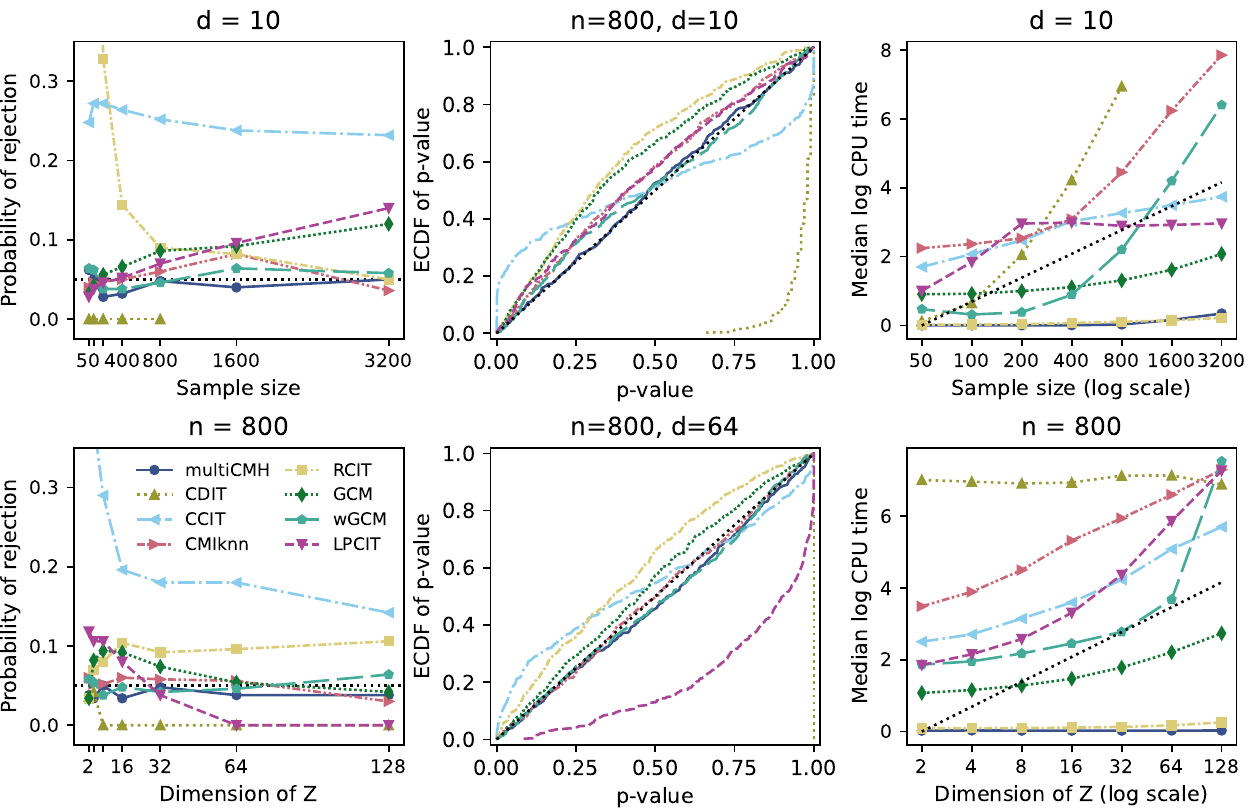}
	\caption{Results of Simulation~1. Probabilities of rejection are truncated at $0.35$, and the nominal level $0.05$ is indicated as a gray dotted line. In the ECDF and median log CPU time plots, a $45^\circ$ reference line through the origin is also drawn in gray dotted line.}
	\label{fig:sim1 T1E}
\end{figure}

\begin{figure}[t!]
	\centering
	\includegraphics[width=0.95\linewidth]{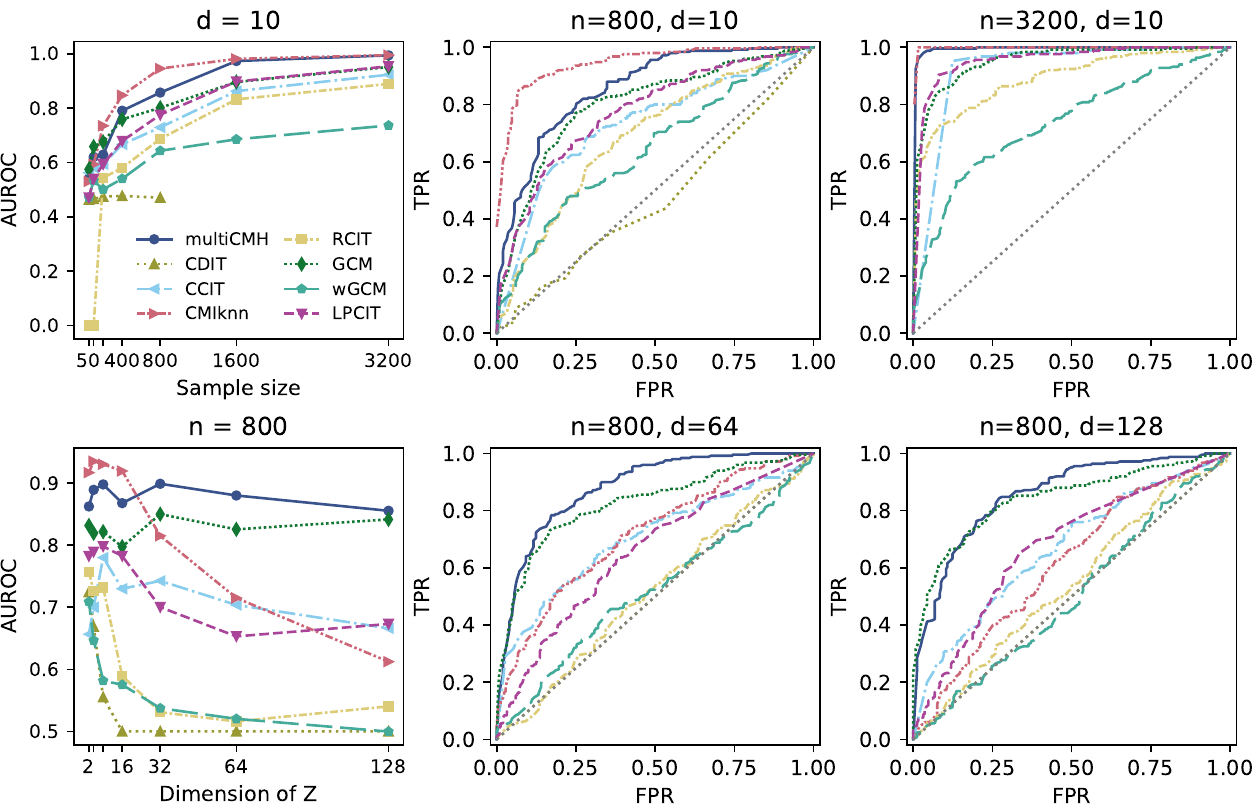}
	\caption{Results of Simulation~1. In the ROC plots, a $45^\circ$ reference line through the origin is included as a gray dotted line.}
	\label{fig:sim1 power}
\end{figure}

For a fair comparison, we adopt a standard data-generating procedure from the literature, referred to as the post-nonlinear noise model \citep{zhang2009identifiability}.
Originally introduced in \citet{zhang2012kernel}, this scenario has since become a benchmark for assessing methodological performance in numerous studies \citep{doran2014permutation, runge2018conditional, strobl2019approximate, scetbon2022asymptotic, li2023k, yang2025conditional}.

Under the null scenario $\mcl{H}_0:X\perp Y\mid Z$, we generate $n$ samples with $d$ conditioning variables.
For each $i$, the conditioning variables $Z_{i1}, \cdots, Z_{id}$ and noise terms $\epsilon_1, \epsilon_2$ are drawn independently from a standard normal distribution $\opn{N}(0,1)$.
The nonlinear functions $f_1$ and $f_2$ are independently selected from the set $\left\{(\cdot), (\cdot)^2, (\cdot)^3, \tanh(\cdot), \exp(-|\cdot|)\right\}$ uniformly at random.
$X_i$ and $Y_i$ are then generated as $X_i = f_1 \left(\sum_{j=1}^{\lfloor d/2\rfloor}Z_{ij}/\lfloor d/2\rfloor + \epsilon_1\right)$ and $Y_i = f_2 \left(\sum_{j=1}^{\lfloor d/2\rfloor}Z_{ij}/\lfloor d/2\rfloor + \epsilon_2\right)$.
In this setup, only the first $\lfloor d/2\rfloor$ components of $Z_i$ are used to generate $X_i$ and $Y_i$.
For the alternative $\mcl{H}_1:X\not\perp Y\mid Z$, we introduce a common latent confounder $\epsilon_3 \sim \opn{N}(0,1)$ and let $X_i = f_1 \left(0.8 \epsilon_3 + \epsilon_1\right)$, $Y_i = f_2 \left(0.8 \epsilon_3 + \epsilon_2 \right)$.
The shared confounder $\epsilon_3$ induces a dependency between $X_i$ and $Y_i$ that persists even after conditioning on the spurious $Z_i$.

To assess T1E control, we generate $500$ datasets under $\mcl{H}_0$ and report, for each method, the proportion of rejections at the significance level $\alpha=0.05$ and the empirical cumulative distribution function (ECDF) of $p$-values.
The former serves as a summary diagnostic, while the latter provides a visual assessment of each method's behavior.
For a well‑calibrated test, the $p$‑values under the null should closely follow, or at least be stochastically larger than, the uniform distribution.
These metrics are computed for a fixed $d=10$ with varying $n \in \{50, 100, 200, 400, 800, 1600, 3200\}$, and for a fixed $n=800$ with varying $d \in \{2, 4, 8, 16, 32, 64, 128\}$.

For power analysis, we emphasize that comparing rejection proportions under $\mcl{H}_1$---a common practice in the literature---can be misleading.
As our results demonstrate, some methods achieve inflated power by failing to adequately control T1E.
To provide a more robust assessment, we instead examine the full receiver operating characteristic (ROC) curve and report the area under the curve (AUROC) as the primary performance metric.
This power analysis is based on 250 datasets generated under $\mcl{H}_0$ and another 250 generated under $\mcl{H}_1$, with the same grid of $(n, d)$ values as above.

Figure~\ref{fig:sim1 T1E} shows that for a fixed dimension of $d=10$, multiCMH successfully controls T1E below the significance level for all sample sizes $n \geq 100$.
The empirical distribution of its $p$‑values is close to uniform, demonstrating rapid convergence to the null distribution.
This control is maintained across all tested dimensions at a sample size of $n=800$.
Because of finite-sample effects, larger conditioning sets generally make $X$ and $Y$ appear conditionally independent, so the T1E of other methods also decreases with increasing dimension---except for CCIT and RCIT, which exhibit a serious miscalibration issue.
Other than multiCMH, CDIT consistently controls T1E under all settings. 
However, its $p$‑values are concentrated near one, indicating conservative behavior that comes at the expense of its power: its AUROC precipitates to $0.5$ for $d\geq 10$ (Figure~\ref{fig:sim1 power}).

Our power analysis in Figure~\ref{fig:sim1 power} verifies that multiCMH is consistent, with its AUROC converging to almost one for sample sizes $n\geq 1600$.
Apart from multiCMH, CMIknn is the only method whose AUROC is close to one at $n=3200$, but its performance degrades sharply in higher dimensions.
In contrast, the power of multiCMH is robust to dimensionality, countering the popular misconception that discretization methods are especially vulnerable to the curse of dimensionality.
Many other methods, except for GCM, exhibit severely decreasing power as the number of conditioning variables grows, limiting their applicability to low-dimensional settings.

Beyond statistical performance, computational efficiency is a crucial consideration.
The rightmost panel of Figure~\ref{fig:sim1 T1E} reports the median CPU time per run, illustrating how computational complexity scales with both $n$ and $d$.
Across all $(n,d)$ pairs, multiCMH is the fastest among the methods; for example, Table~\ref{tab:timecomp} shows the median CPU time of each method for $n=800$ and $d=128$.
CDIT, CMIknn, and wGCM display polynomial growth in $n$; for instance, CDIT requires more than 1000 seconds for $n=800$ and fails to complete for larger sample sizes.
Furthermore, the runtime of multiCMH remains stable as $d$ increases, whereas LPCIT and wGCM exhibit apparent polynomial growth in dimensionality.

\begin{table}[t!]
\centering
\begin{tabular}{c c c c c c c c}
\hline
multiCMH & RCIT & GCM & CCIT & CDIT & LPCIT & CMIknn & wGCM \\
\hline
1.03 & 1.29 & 15.43 & 299.56 & 978.51 & 1435.83 & 1470.74 & 1864.53 \\
\hline
\end{tabular}
\caption{Median CPU seconds for $n=800$ and $d=128$, rounded to two decimal places.}
\label{tab:timecomp}
\end{table}

\subsection{Simulation 2: computational scalability}
In this simulation, we assess the scalability of the methods on large datasets.
We generate data under the null hypothesis, where all variables are drawn independently from a standard normal distribution, and we consider large sample sizes ranging from one thousand to over one million ($n \in \{2^i\times 1000: i=0,1,\cdots, 10\}$) and dimensions $d \in \{1, 10, 100\}$.
We repeat each experiment $30$ times and report the median log CPU times.
Methods that exhibited polynomial time complexity in $n$ or $d$ during the previous simulation (namely, CDIT, CMIknn, and wGCM) are excluded, as they are not computationally feasible at this scale.
Each method was allocated at most 64~GB of RAM.

Figure~\ref{fig:sim2} confirms that the time complexity of multiCMH is nearly linear.
This is expected, as the primary computational bottleneck is the sorting step required for partitioning and stratification, which has a complexity of $O(n\log n)$.
The figure also shows that the runtime barely increases with $d$, because our stratification algorithm performs one-dimensional sorting by looping over the axes.
Beyond speed, multiCMH is memory-efficient, as its memory footprint for any dataset is dominated by storing a fixed number of $2\times 2\times T$ contingency tables.
In contrast, other methods (except RCIT) require significantly more memory and fail to run on datasets beyond a certain scale.

\begin{figure}[t!]
	\centering
	\includegraphics[width=0.85\linewidth]{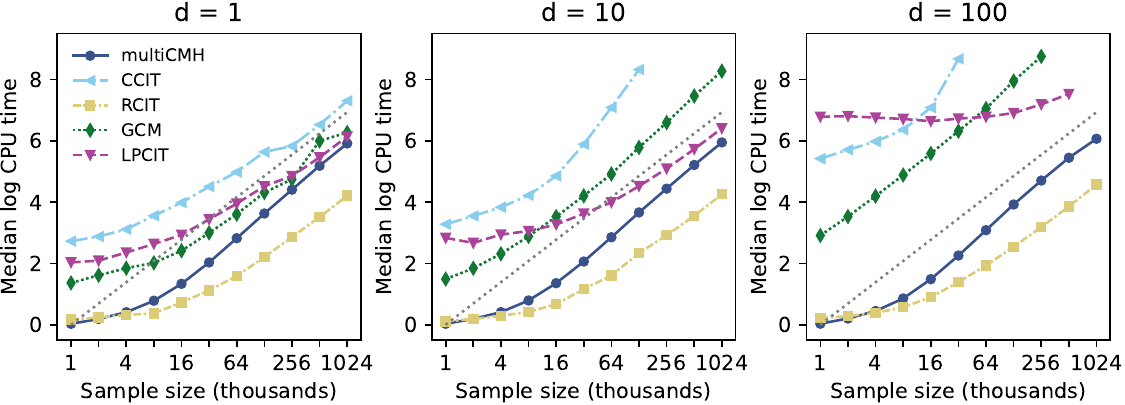}
	\caption{Results from Simulation 2. The gray dotted line is a $45^\circ$ reference. The results of each method are plotted only if they could be executed within the 64~GB RAM budget.}
	\label{fig:sim2}
\end{figure}

\section{Case study: Uber ride-share request data}

We apply our method to rider session data from the Uber ride-share platform, collected from an anonymous U.S. metropolitan area over multiple days. When a rider opens the Uber app and enters a destination, they are presented with ride options (e.g., UberX, Comfort). For each option, the app displays several key pieces of information, including the estimated pickup time (how long until a driver arrives) and the upfront fare (the total price of the trip). Based on this information, the rider decides whether to request a ride or abandon the session. Understanding which factors influence conversion at this stage has direct implications for marketplace operations, including driver positioning, pricing strategy, and product design.

Our dataset captures rider sessions at this decision point. Each observation records: (1) Pickup ETA: The estimated wait time for driver arrival, (2) Upfront Fare: The quoted price for the trip, inclusive of any dynamic pricing adjustments, and (3) Pricing Condition: Indicators of the pricing environment (e.g., baseline vs.\ elevated pricing periods). The outcome variable $Y$ is binary: whether the rider requested a ride ($Y = 1$) or exited without requesting ($Y = 0$). Note: Variable definitions and value ranges have been generalized for confidentiality.

We test the null hypothesis $X\perp Y \mid Z$, where $X$ is ETA, $Y$ the trip request decision ($Y=1$ if a request is made), and $Z$ the pricing context (price level and pricing conditions). The key question is: does pickup wait time affect a rider's decision to request, even after accounting for price? If riders are purely price-sensitive, ETA should have no residual association with conversion once we condition on fare. If riders also value early pickups, we expect to see a negative association between ETA and request probability, even at fixed price levels. Conditioning on price is important because ETA and price are often correlated through marketplace dynamics. During periods of high demand and limited supply, both wait times and prices tend to increase. Without conditioning, any observed ETA-conversion relationship could be confounded by this shared dependence on supply-demand balance.

In our analysis, the overall corrected $p$-value is numerically evaluated as zero, indicating that it lies below the lower bound of machine precision. This is not surprising. At the sample sizes typical in many observational studies of industry applications, it is rare to observe true conditional independence unless all potential confounders are accounted for. Therefore, any consistent testing procedure would reject the null hypothesis in our case. More interesting questions are: Where in the covariate space is the dependency strongest? How does the effect vary across pricing conditions?

To character the detected dependencies, we report the estimated common log odds ratio. Specifically, in terms of the $2\times 2\times T$ table constructed in a significant window $I\times J$ given a stratification $\mcl{S}$, an estimator of its common log odds ratio \citep{mantel1959statistical}
\begin{align*}
    \hat{\theta}_{\mcl{S}} = \log \frac{
        \sum_{S\in \mcl{S}}n(I^{left}, J^{left},S)n(I^{right}, J^{right},S)/n(I, J,S)
    }{
        \sum_{S\in \mcl{S}}n(I^{left}, J^{right},S)n(I^{right}, J^{left},S)/n(I, J,S)
    },
\end{align*}
provides a measure of the strength and direction of the conditional association within the window.
An estimator of its variance, denoted $\hat{\sigma}^2(\hat{\theta}_{\mcl{S}})$, is also available under the sparse-data asymptotics of a growing number of strata \citep{phillips1987estimators}.
Note that $\exp(\hat{\theta}_{\mcl{S}})$ is a weighted average of $\exp(\hat{\theta}(S))$, where $\hat{\theta}(S)$ denotes a stratum-specific sample log odds ratio \citep[p.~40]{plackett1974categorical}:
\begin{align*}
    \hat{\theta}(S) = \log \frac{
        (n(I^{left}, J^{left},S) + 0.5) (n(I^{right}, J^{right},S)+0.5)
    }{
        (n(I^{left}, J^{right},S)+0.5) (n(I^{right}, J^{left},S)+0.5)
    }, \quad S \in \mcl{S}.
\end{align*}
Following the rejection of the null, the empirical distributions of $\hat{\theta}(S)$ in the significant windows across levels of conditioning variables reveal how the strength and direction of the association vary with $Z$.

\begin{table}[t!]
    \centering
    \small
    \begin{tabular}{c c c c c c c c}
    \toprule
    $(l_1,l_2)$ & $I\times J$ & $\alpha_n(I,J)$ & $p(I,J,\mcl{S}_{IJ})$ & $\hat{\theta}_{\mcl{S}}$ & $\hat{\sigma}(\hat{\theta}_{\mcl{S}})$ & $2.5\%$ & $97.5\%$ \\
    \midrule
    $(0,0)$ & $(0,1]\times[0,1]$     & $7.3\times 10^{-3}$ & $0.0\times 10^{0}$   & $-0.10$ & $0.01$ & $-0.11$ & $-0.08$ \\
    $(1,0)$ & $(0.24,1]\times[0,1]$  & $3.7\times 10^{-3}$ & $4.5\times 10^{-10}$ & $-0.06$ & $0.01$ & $-0.07$ & $-0.04$ \\
    $(4,0)$ & $(0.47,1]\times[0,1]$  & $7.3\times 10^{-3}$ & $3.4\times 10^{-4}$  & $-0.10$ & $0.03$ & $-0.15$ & $-0.04$ \\
    $(5,0)$ & $(0.53,1]\times[0,1]$  & $7.3\times 10^{-3}$ & $2.1\times 10^{-6}$  & $-0.18$ & $0.04$ & $-0.25$ & $-0.10$ \\
    $(6,0)$ & $(0.59,1]\times[0,1]$  & $7.3\times 10^{-3}$ & $2.0\times 10^{-3}$  & $-0.16$ & $0.05$ & $-0.27$ & $-0.06$ \\
    \bottomrule
    \end{tabular}
    \captionof{table}{Significant windows when $X$ is ETA (rescaled to $(0,1]$), and $Y\in \{0,1\}$ is trip request ($Y=1$ if a request is made). 
    For each window $I\times J$ where $I\in \mcl{I}_{l_1}$ and $J\in \mcl{I}_{l_2}$,
    $\alpha_n(I,J)$ denotes the Šidák-corrected significance level, 
    while $p(I,J,\mcl{S}_{IJ})$ is the raw $p$-value prior to correction.
    $\hat{\theta}_{\mcl{S}}$ and $\hat{\sigma}(\hat{\theta}_{\mcl{S}})$ represent the estimated common log odds ratio for the window and its standard error, respectively.
    Confidence intervals are reported as $\hat{\theta}_{\mcl{S}}\pm 1.96 \hat{\sigma}(\hat{\theta}_{\mcl{S}})$ following \citet[p.~230]{agresti2013categorical}.
    }
    \label{tab:analysis1}
\end{table}

\begin{figure}[t!]
    \includegraphics[width=\linewidth]{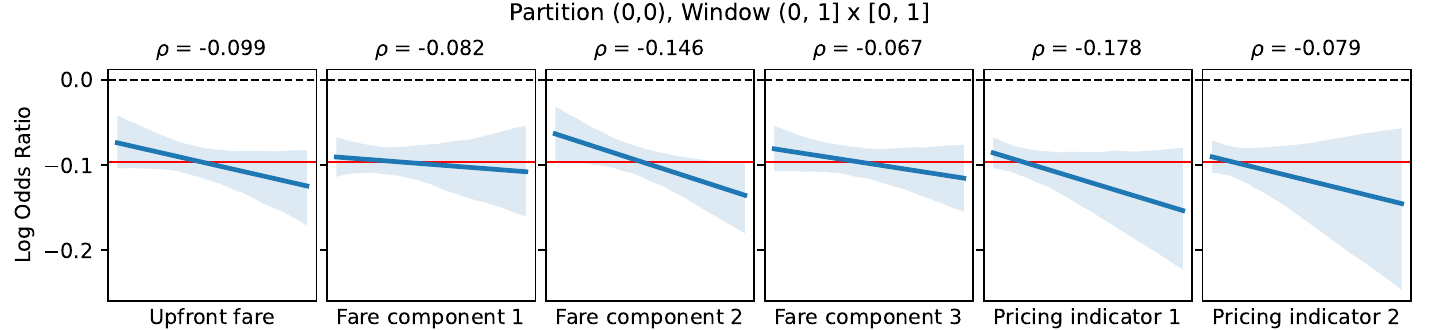}
    \vspace{0.001cm}
    
    \includegraphics[width=\linewidth]{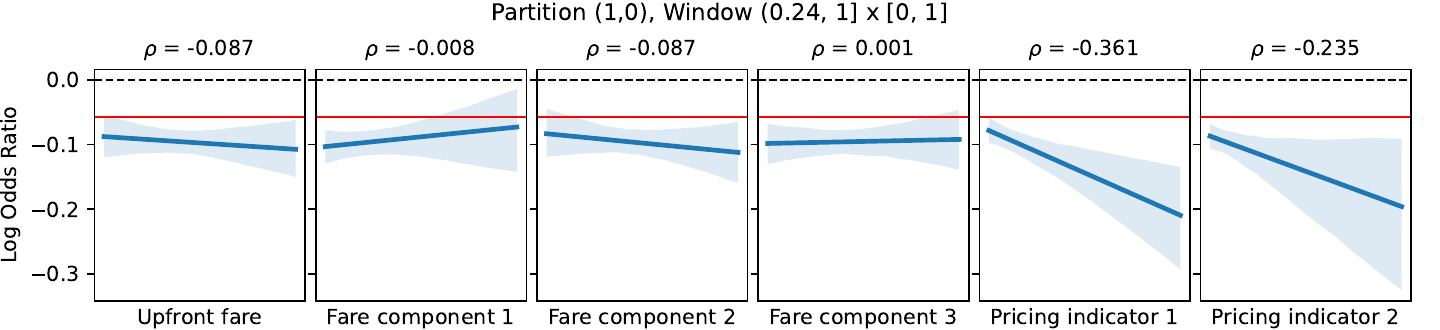}
    \captionof{figure}{For each window (row), the empirical distributions of $\hat{\theta}(S)$---the stratum-specific sample log odds ratios---are plotted against the stratum-specific sample means of the conditioning variables on $x$-axes.
    Factors related to the fare components include time, distance, and base fare, among others.
    The common log odds ratio $\hat{\theta}_\mcl{S}$ of the window is indicated by the red line.
    For confidentiality reasons, $x$-axes tick marks and individual points are not displayed.
    Instead, fitted regression lines (blue) and their associated $95\%$ confidence bands (sky-blue shading) illustrate how $\hat{\theta}(S)$ varies with the conditioning variables. 
    $\rho$ denotes Pearson's correlation coefficient.}
    \label{fig:analysis1}
\end{figure}

Table~\ref{tab:analysis1} reports the summary statistics for significant windows. The $95\%$  confidence intervals of the estimated common log odds ratios across all significant windows lie strictly below zero, confirming our conjecture that a longer ETA discourages users from requesting a ride. More substantively, the magnitude of this effect is not uniform. Windows covering the upper range of ETA (e.g., $(0.53, 1]$ and $(0.59, 1]$) exhibit larger negative log odds ratios ($-0.18$ and $-0.16$) compared to windows that include shorter ETAs ($-0.06$ to $-0.10$), suggesting that rider sensitivity to wait time intensifies at longer ETAs.

Figure~\ref{fig:analysis1} shows an additional dimension of heterogeneity: the strength of the ETA-conversion relationship varies systematically with the conditioning variables. This variation is particularly evident in the upper ETA range, suggesting that pricing context changes rider sensitivity to wait times. Such interaction structure is precisely the type of insight that the multiscale framework is designed to uncover.

\section{Concluding remarks}
We conclude this work with a remark on applying our framework to multivariate $X$ and $Y$, a particularly pressing issue in modern applications where $X$ may be a high-dimensional random object.
Our method does not directly use each observed data point; instead, it leverages the information encoded by its path from the root to the leaf node in the nested dyadic partition tree.
This process effectively compresses each data point $x_i$ into a binary sequence of $\{0,1\}$, where $0$ indicates membership in the left child and $1$ in the right child.
In this way, any sample space---including $\mbb{R}^d$---is encoded as an integer represented in binary according to the nested dyadic partition tree.

While this compression inevitably discards some information, it retains sufficient structure to characterize CI.
Definition~\ref{def:k1k2condind} formalizes CI of these encoded binary sequences, and Theorem~\ref{thm:condind} establishes its equivalence with CI in the original space.
Theorem~\ref{thm:smalltobig} motivates a practical algorithm to test CI based on this formulation, namely our proposed multiCMH, while our stratification algorithm constructs nested dyadic partitions of multivariate $\mcl{X}$, ensuring that the left and right child nodes contain equal numbers of observations.
Nevertheless, detecting the conditional associations in high-dimensional spaces is inherently challenging due to data sparsity.
Our future work will focus on increasing the power of our test in this setting.
One possible remedy is a data-adaptive approach, similar to \citet{gorsky2022multi}, in which the order of dimension splits in our stratification procedure is determined by the strength of the conditional associations revealed in the resulting tables.

\section{Acknowledgment}

This research is partly supported by NSF grant DMS-2152999. Part of the research was carried out when LM was at Duke University.

\bibliographystyle{plainnat}
\bibliography{references}

\clearpage 
\appendix 
\counterwithin{equation}{section} 
\section*{Supplementary Materials} 
\setcounter{section}{0} 
\renewcommand{\thesection}{S\arabic{section}} 
\renewcommand{\theequation}{S\arabic{equation}} 
\renewcommand{\thetheorem}{S\arabic{theorem}}
\renewcommand{\thelemma}{S\arabic{lemma}}
\renewcommand{\thefigure}{S\arabic{figure}}
\setcounter{figure}{0} 

\allowdisplaybreaks

\section{Proofs of the results in the main text} \label{Section: proofs}

\subsection{Proof of Theorem~\ref{thm:tbl22 T1E}}\label{Section: Proof of tbl22 T1E}
The proof is similar to that of Theorem 2 in \cite{kim2022local}. The key is the following lemma, which is a restatement of Lemma 1 in \cite{kim2022localsupp}:
\begin{lemma}\label{supp-lem:thm1-1}
	Let $\mcl{S}=\{S_t:t\in[T]\}$ be a $T$-stratification of $\mcl{Z}$.
	If $P\in \mcl{P}_0$, then
	\begin{align*}
		 & \mcl{D}_H^2 \left(
		\mbb{E}_{Z_t} P_{XY\mid Z}(X,Y\mid Z_t),
		\mbb{E}_{Z_t} P_{X\mid Z}(X\mid Z_t)
		\mbb{E}_{Z_t} P_{Y\mid Z}(Y\mid Z_t)
		\right)               \\ &\leq
		6\mbb{E}_{Z_t, Z_t', Z_t''} \left[
			\mcl{D}_H^2\left(
			P_{X\mid Z}(X\mid Z_t), P_{X\mid Z}(X\mid Z_t')
			\right)
			\mcl{D}_H^2\left(
			P_{Y\mid Z}(Y\mid Z_t), P_{Y\mid Z}(Y\mid Z_t'')
			\right)
			\right]
	\end{align*}
	where $Z_t, Z_t', Z_t''$ are iid samples of $P_Z(Z\mid Z\in S_t)$.
\end{lemma}
\noindent The proof can be found in \cite{kim2022localsupp}.

Consider any marginally smooth $P \in \mcl{P}_0$.
Given a stratification $\mcl{S}$, let $\tilde{P}$ be its discretization defined as \eqref{eqn:tildep} and $\tilde{P}_0$ its CI projection as in \eqref{eqn:tildep0}.
Let $\tilde{p}$ and $\tilde{p}_0$ be the corresponding densities with respect to $\mu$.
Since $\psi_n=1_{M_n^2 > \chi^2_{\alpha,1}}$ is a function of the cell counts $n(x, y, S_t)$ for $S_t \in \mcl{S}$, we only need to consider $\tilde{P}$ instead of $P$.
The T1E of $\psi_n$ under $\tilde P^n$ can be decomposed as
\begin{equation*}
	\mbb{E}_{\tilde{P}^n}[\psi_n] =
	\int \psi_n (\tilde{p}_0^n-\tilde{p}_0^n+\tilde{p}^n)d\mu \leq
	\mbb{E}_{\tilde{P}_0^n}[\psi_n] + 2\mcl{D}_{TV}(\tilde{P}^n_0.\tilde{P}^n)
\end{equation*}
where $\mcl{D}_{TV}(P,Q) = \int |p-q| d\mu/2$ is the TV distance.

For the first term, since by the construction $\tilde{P}_0$ belongs to $\tilde{\mcl{H}}_0$, the conditional distribution of $M_n$ given margin totals $n(\cdot, y, S_t)$ and $n(x, \cdot, S_t)$ for all strata asymptotically follows $\chi^2_1$, hence $\lim_{n\to \infty}\mbb{E}_{\tilde{P}^n}[\psi_n]\leq \alpha$.
For the second term, from well-known inequalities between the TV and Hellinger distances, it can be inferred that for any distributions $P_1$ and $P_2$ with densities with respect to a common base measure, $\mcl{D}_{TV}(P_1^n,P_2^n)\to 0$ if and only if $n\mcl{D}_H^2(P_1,P_2)\to 0$. Therefore, we address the Hellinger distance between $\tilde P$ and $\tilde P_0$ as follows (note that $s(t) = P_Z\{Z\in S_t\}$):
\begin{align*}
	 & \mcl{D}_{TV}(\tilde{P}^n,\tilde{P}_0^{n})
	\leq \sqrt{2}\mcl{D}_{H}(\tilde{P}^n,\tilde{P}_0^{n})
	\leq \sqrt{2n}\mcl{D}_{H}(\tilde{P},\tilde{P}_0) \\
	 & \leq
	\sqrt{n}\left\{
	\sum_{t=1}^T s(t)\sum_{X,Y}
	\left(
	\sqrt{\mbb{E}_{Z_t} P_{XY\mid Z}(X,Y\mid Z_t)} -
	\sqrt{\mbb{E}_{Z_t} P_{X\mid Z}(X\mid Z_t)
			\mbb{E}_{Z_t}P_{Y\mid Z}(Y\mid Z_t)}
	\right)^2
	\right\}^{1/2}                                   \\
	 & =
	\sqrt{2n}\bigg\{
	\sum_{t=1}^T s(t)
	\mcl{D}_H^2\Big(
	\mbb{E}_{Z_t}P_{XY\mid Z}(X,Y\mid Z_t),
	\mbb{E}_{Z_t}P_{X\mid Z}(X\mid Z_t) \mbb{E}_{Z_t} P_{Y\mid Z}(Y\mid Z_t)
	\Big)
	\bigg\}^{1/2}                                    \\
	 & \stackrel{(a)}{\leq}
	\sqrt{12n}\bigg\{
	\sum_{t=1}^T s(t)
	\mbb{E}_{Z_t,Z_t',Z_t''}
	\left[
		\mcl{D}_H^2
		\left( P_{X\mid Z}(X\mid Z_t), P_{X\mid Z}(X\mid Z_t')\right) 
		\mcl{D}_H^2\left( P_{Y\mid Z}(Y\mid Z_t), P_{Y\mid Z}(Y\mid Z_t'')\right)
		\right]
	\bigg\}^{1/2}                                    \\
	 & \stackrel{(b)}{\leq}
	\sqrt{12n}\bigg\{
	\sum_{t=1}^T s(t)
	\mbb{E}_{Z_t,Z_t',Z_t''}
	\left[ L^2h^2 \cdot L^2h^2 \right]
	\bigg\}^{1/2}                                    \\
	 & = Cn^{1/2}h^2
\end{align*}
where $(a)$ follows from Lemma \ref{supp-lem:thm1-1} and $(b)$ is due to $P$ being marginally smooth. Therefore, if $h=o(n^{-1/4})$, then $\mcl{D}_{TV}(\tilde{P}^n,\tilde{P}_0^{n})\to 0$ as $n\to \infty$, which completes the proof.
\qed

\subsection{Proof of Theorem~\ref{thm:tbl22 Power}}\label{Section: Proof of tbl22 Power}
We need the following two lemmas. The first one allows us to bound the difference between the marginal and conditional log odds ratios as a function of $h$.

\begin{lemma}\label{supp-lem:thm1-2}
	Let $\mcl{X}=\mcl{Y}=\{0,1\}$ and $\mcl{Z}=\mbb{R}^d$.
	Let $P$ be a distribution of $(X,Y,Z)$, of which $P(X,Y \mid Z)$ is continuous in $Z$.
	Let $\mcl{S}=\{S_t:t\in[T]\}$ be a $T$-stratification of $\mcl{Z}$ and $h=\max_{S\in \mcl{S}}\sup_{z,z'\in S}\delta(z,z')$.
	Suppose that $P$ satisfies:
	\begin{enumerate}[(i)]
		\item $\min_{X,Y} P_{XY\mid Z}(X,Y\mid Z) > p_{min}$ $P$-a.s. for some $0<p_{min}<1$,
		\item for all $z, z'\in \mcl{Z}$,
		      \begin{equation*}
			      \left|P_{X\mid Z}(X\mid z) - P_{X\mid Z}(X\mid z')\right|
			      \vee
			      \left|P_{Y\mid Z}(Y\mid z) - P_{Y\mid Z}(Y\mid z')\right|
			      \leq L\delta(z,z'),
		      \end{equation*}
	\end{enumerate}
	For any stratum $S_t$, select any pair $(X,Y)$ and choose $z_t^\ast \in \opn{cl}(S_t)$ such that $\mbb{E}_{Z_t} P_{XY\mid Z}(X,Y \mid Z_t) = P_{XY\mid Z}(X,Y\mid z_t^\ast)$ where $Z_t\sim P_{Z}(Z\mid Z\in S_t)$ and $\opn{cl}(S_t)$ is a closure of $S_t$.
	Then for small enough $h<p_{min}/L$, we have $|\theta_t-\theta(z_t^\ast)|=O(h)$.
\end{lemma}
\begin{proof}
	Without loss of generality, choose $z_t^\ast \in \opn{cl}(S_t)$ such that. $\mbb{E}_{Z_t}P_{XY\mid Z}(0,0\mid Z_t) = P(0,0\mid z_t^\ast)$, which always exists since $P(X,Y\mid Z)$ is continuous in $Z$ and $\opn{cl}(S_t)$ is compact and connected.
	For brevity, we substitute $P$ for $P_{XY\mid Z}$ when there is no confusion. Then
	\begin{align*}
		e^{\theta_t} & =
		\frac{\mbb{E}_{Z_t}P(0,0\mid Z_t) \mbb{E}_{Z_t}P(1,1\mid Z_t)}
		{\mbb{E}_{Z_t}P(0,1\mid Z_t)\mbb{E}_{Z_t}P(1,0\mid Z_t)} \\&=
		\frac{P(0,0\mid z_t^\ast) P(1,1\mid z_t^\ast)}
		{P(0,1\mid z_t^\ast)P(1,0\mid z_t^\ast)}
		\frac{\mbb{E}_{Z_t}P(1,1\mid Z_t)}{P(1,1\mid z_t^\ast)}
		\frac{P(0,1\mid z_t^\ast)}{\mbb{E}_{Z_t}P(0,1\mid Z_t)}
		\frac{P(1,0\mid z_t^\ast)}{\mbb{E}_{Z_t}P(1,0\mid Z_t)}  \\&=
		e^{\theta(z_t^\ast)}
		\frac{\mbb{E}_{Z_t}P(1,1\mid Z_t)}{P(1,1\mid z_t^\ast)}
		\frac{P(0,1\mid z_t^\ast)}{\mbb{E}_{Z_t}P(0,1\mid Z_t)}
		\frac{P(1,0\mid z_t^\ast)}{\mbb{E}_{Z_t}P(1,0\mid Z_t)}.
	\end{align*}
	Let $z_t^X\in \opn{cl}(S_t)$ s.t. $\opn{E_{Z_t}} P_{X\mid Z}(0\mid Z_t)=P_{X\mid Z}(0\mid z_t^X)$ and $\Delta_X = P_{X\mid Z}(0\mid z_t^X)-P_{X\mid Z}(0\mid z_t^\ast)$.
	Similarly, let $z_t^Y\in \opn{cl}(S_t)$ s.t. $\opn{E_{Z_t}} P_{Y\mid Z}(0\mid Z_t)=P_{Y\mid Z}(0\mid z_t^Y)$ and $\Delta_Y = P_{Y\mid Z}(0\mid z_t^Y)-P_{Y\mid Z}(0\mid z_t^\ast)$.
	Then
	\begin{align*}
		\mbb{E}_{Z_t}P(0,1\mid Z_t) & =
		\mbb{E}_{Z_t}\left[ P_{X\mid Z}(0 \mid Z_t) - P(0,0\mid Z_t) \right]                      \\&=
		P_{X\mid Z}(0\mid z_t^X)-\mbb{E}_{Z_t} P(0,0\mid Z_t)                                     \\&=
		P_{X\mid Z}(0\mid z_t^X)-P(0,0\mid z_t^\ast)                                              \\&=
		P_{X\mid Z}(0\mid z_t^X) - \left(P_{X\mid Z}(0\mid z_t^\ast) - P(0,1\mid z_t^\ast)\right) \\&=
		\Delta_X+ P(0,1\mid z_t^\ast),
	\end{align*}
	and $\mbb{E}_{Z_t}P(1,0\mid Z_t)$, $\mbb{E}_{Z_t}P(1,1\mid Z_t)$ can be expressed in a similar manner. Hence
	\begin{align*}
		\theta_t & = \theta(z_t^\ast) +
		\log \left(1 - \frac{\Delta_X}{P(1,1\mid z_t^\ast)}- \frac{\Delta_Y}{P(1,1\mid z_t^\ast)}\right) \\&\quad-
		\log \left(1 + \frac{\Delta_X}{P(0,1\mid z_t^\ast)}\right) -
		\log \left(1 + \frac{\Delta_Y}{P(1,0\mid z_t^\ast)}\right).
	\end{align*}
	For simplicity, denote $P(1,1\mid z_t^\ast)$ by $p_{11}$, and similarly for other probabilities.
	By condition (ii), $\Delta_X \vee \Delta_Y \leq Lh$ and $p_{11}\geq \min_{X,Y}P(X,Y\mid z_t^\ast) > p_{min}$ imply that the ratio $\Delta_X/p_{11}$ can be made arbitrarily close to zero with a smaller $h$.
	Since $\log(1+x) = x - x^2/2 +x^3/3 - \cdots = x + o(x)$ around zero,
	\begin{align*}
		|\theta_t - \theta(z_t^\ast)| & \leq
		\frac{\Delta_X}{p_{11}} + \frac{\Delta_Y}{p_{11}} +
		\frac{\Delta_X}{p_{01}} + \frac{\Delta_Y}{p_{10}} +
		o\left(\frac{\Delta_X}{p_{11}} + \frac{\Delta_Y}{p_{11}}\right) +
		o\left(\frac{\Delta_X}{p_{01}}\right) + o\left(\frac{\Delta_Y}{p_{10}}\right) \\
		                              & \leq 4Ch + o(h)
	\end{align*}
	where $C=L/p_{min}$.
	Therefore, $|\theta_t-\theta(z_t^\ast)|=O(h)$.
\end{proof}

The second lemma bounds the squared Hellinger distance within a class of one-parameter family of distributions in terms of the difference in the parameter. 
This lemma is adapted from \citet[Theorem~7.13]{polyanskiy_fdiv}.

\begin{lemma}\label{supp-lem:thm1-3}
	Let $\{P_t:t\in[a,b)\}$ be a family of distributions with densities $p_t(x)$ with respect to a finite base measure $\mu$. If $h_t(x)=\sqrt{p_t(x)}$ satisfies all of the following:
	\begin{enumerate}[(i).]
		\item $\int_a^b |h_s'(x)|ds < \infty$ and $\lim_{t\downarrow a} h_t'(x)=h'(a)$ $\mu$-almost surely where $h_t'(x)=dh_t(x)/dt$ so that $h_t(x)$ can be written as
		      \begin{equation*}
			      h_t(x) = h_a(x) + \int_a^t h_s'(x)ds,\quad t\in[a,b),
		      \end{equation*}
		\item $\{[h_t'(x)]^2:t\in [a,b)\}$ is uniformly $\mu$-integrable.
	\end{enumerate}
	Then, as $t\to a$ for $t\in [a,b)$,
	\begin{equation*}
		\mcl{D}_H^2(P_t,P_a) = J(a)(t-a)^2/8 + o((t-a)^2)
	\end{equation*}
	where $J(a) = \int [p_a'(x)]^2/p_a(x)d\mu <\infty$ is the Fisher information of $P_t(x)$ at $t=a$.
\end{lemma}

\begin{proof}
	Note that $2\mcl{D}_H^2(P_t,P_a)=\int (h_t-h_a)^2d\mu$ where
	\begin{equation*}
		h_t(x) = h_a(x) + \int_a^t h_s'(x)ds = h_a(x) + \int_0^1 (t-a)h'_{v(u)}(x) du
	\end{equation*}
	where $v(u)=a+(t-a)u$. Hence
	\begin{align}
		\frac{2}{(t-a)^2}\mcl{D}_H^2(P_t,P_a) & =\int \left(\int_0^1 h'_{v(u)}(x) du\right)^2 d\mu \notag                                     \\
		                                      & =\int \int_0^1 h'_{v(u_1)}(x) du_1 \int_0^1 h'_{v(u_2)}(x) du_2 d\mu \notag                   \\
		                                      & = \label{supp-eqn:lem3pf1}\int \int_0^1 \int_0^1 h'_{v(u_1)}(x) h'_{v(u_2)}(x) du_1du_2 d\mu.
	\end{align}
	Note that $h'_{v(u)}$ is integrable with respect to $u\in [0,1)$ since $\int_a^b |h_s'(x)|ds < \infty$. Moreover, by the Cauchy–Schwarz inequality,
	\begin{equation*}
		\int h'_{v(u_1)} h'_{v(u_2)} d\mu \leq \int [h'_{v(u)}]^2d\mu <\infty
	\end{equation*}
	which is finite by the assumption that $\{[h_t'(x)]^2:t\in [a,b)\}$ is uniformly $\mu$-integrable.
	Therefore, we can apply Fubini's Theorem to \eqref{supp-eqn:lem3pf1}, writing $H\left(v(u_1), v(u_2)\right) = \int h'_{v(u_1)} h'_{v(u_2)} d\mu$,
	\begin{equation*}
		\frac{2}{(t-a)^2}\mcl{D}_H^2(P_t,P_a) =
		\int_0^1\int_0^1 H\left(v(u_1), v(u_2)\right) du_1du_2.
	\end{equation*}
	Now we investigate the limit of the integrand as $t\to a$. By the assumptions, $\{h_{t_1}'(x)h_{t_2}'(x):t_1,t_2\in [a,b)\}$ is uniformly $\mu$-integrable and $\lim_{t_1,t_2\downarrow a}h'_{t_1}h'_{t_2}(x)=h_a'(x)^2$ $\mu$-almost surely. Since $\mu$ is a finite measure, pointwise $\mu$-almost sure convergence implies convergence in measure; hence, we can apply Vitali's Convergence Theorem to show convergence in $L^2$:
	\begin{equation*}
		\lim_{t\downarrow a} H\left(v(u_1), v(u_2)\right)=
		\int [h'_a]^2 d\mu = \int [p_a^{-1/2}p_a'/2]^2d\mu
		=\frac{1}{4} \int \frac{(p_a')^2}{p_a} d\mu=J(a)/4.
	\end{equation*}
	Moreover, $H\left(v(u_1), v(u_2)\right)$ is bounded because $\{h_{t_1}'(x)h_{t_2}'(x):t_1,t_2\in [a,b)\}$ is uniformly $\mu$-integrable, so we can apply the Dominated Convergence Theorem to write
	\begin{equation*}
		\lim_{t\downarrow a} \int_0^1\int_0^1 H\left(v(u_1), v(u_2)\right) du_1du_2 =
		\frac{1}{4}\int_0^1\int_0^1 J(a) du_1du_2 = \frac{1}{4}J(a)
	\end{equation*}
	which completes the proof.
\end{proof}

Now we prove Theorem~\ref{thm:tbl22 Power}.
We first prove the case where $Q$ satisfies condition (ii) of Theorem~\ref{thm:tbl22 Power}.
By the assumption, we can partition $\mcl{Z}$ into $\mcl{Z}=U_Q\cup V_Q$, where $U_Q$ is a non-null set such that $\theta(z)>0$ for $^\forall z\in U_Q$ and $\theta(z)\geq 0$ for $^\forall z \in V_Q$.
Let $\mcl{S}=\{S_t:t\in [2T]\}$ be a partition of $\mcl{Z}$ such that $U_Q=\cup_{t=1}^T S_t$ and $V_Q = \cup_{t=T+1}^{2T} S_t$.
For each $S_t$, define $z_t^\ast \in \opn{cl}(S_t)$ as the value satisfying $\mbb{E}_{Z_t}Q_{XY\mid Z}(0,0\mid Z_t) = Q_{XY\mid Z}(0,0\mid z_t^\ast)$, where $\opn{cl}(S_t)$ is the closure of $S_t$.
Such a value exists since $Q(X,Y\mid Z)$ is continuous, and $\opn{cl}(S_t)$ is compact and connected.
We note that, by the assumption of $T\asymp n$, the total number of cell counts in each $2\times 2$ table is bounded above as $n\to \infty$.
Define $Q^\ast$ as a distribution of $(X,Y,Z)$ given by $Q^\ast(X,Y,Z)=Q_Z(Z) Q_{XY\mid Z}^{\ast}(X,Y\mid Z)$ where
\begin{equation*}
	Q_{XY\mid Z}^{\ast}(X,Y\mid Z)=\sum_{t=1}^{2T} Q_{XY\mid Z}(X,Y\mid z_t^\ast) 1_{Z\in S_t}.
\end{equation*}
$Q^\ast$ is a distribution where the marginal and conditional log odds ratios coincide in terms of the stratification $\mcl{S}$.
Let $g_\theta(x\mid n_{0\cdot}, n_{\cdot 0}, n_{\cdot\cdot})$ be the pmf of the non-central Fisher's hypergeometric distribution with the log odds ratio $\theta$:
\begin{equation*}
	g_\theta(x\mid n_{\cdot0}, n_{0\cdot}, n_{\cdot\cdot}) =
	{n_{0\cdot} \choose x} {n_{\cdot\cdot} - n_{0\cdot} \choose n_{\cdot0} - x} e^{x \theta} \Bigg/
	\sum_{y \in \Gamma}
	{n_{0\cdot} \choose y} {n_{\cdot\cdot} - n_{0\cdot} \choose n_{\cdot0} - y} e^{y \theta}
\end{equation*}
where $\Gamma= [0\vee n_{0\cdot}+n_{\cdot0}-n_{\cdot\cdot}, n_{\cdot0}\wedge n_{0\cdot}]$.
We denote $n_{\mcl{S}}(x,y) = \{n(x, y,S_t):S_t \in \mcl{S}\}$ as the set of cell counts of $2\times 2$ tables, $n_{\mcl{S}}(x,\cdot) = \{n(x, \cdot,S_t):S_t \in \mcl{S}\}$, $n_{\mcl{S}}(\cdot,y) = \{n(\cdot, y,S_t):S_t \in \mcl{S}\}$ as sub-table margin totals given a stratification $\mcl{S}$.
Under $Q$, the conditional distribution of $n_{\mcl{S}}(0,0)$ given $n_{\mcl{S}}(0,\cdot), n_{\mcl{S}}(\cdot, 0), n_{\mcl{S}}(\cdot,\cdot)$ is a $2T$-product of non-central hypergeometric distributions with marginal log odds ratios $\{\theta_t : t\in [2T]\}$:
\begin{equation*}
	g_{Q}\left(n_{\mcl{S}}(0,0) \mid n_{\mcl{S}}(0,\cdot), n_{\mcl{S}}(\cdot,0), n_{\mcl{S}}(\cdot,\cdot)\right) =
	\prod_{t=1}^{2T} g_{\theta_t}\left(n(0,0,S_t) \mid n(0,\cdot,S_t), n(\cdot,0,S_t), n(\cdot,\cdot,S_t)\right).
\end{equation*}
Although $\theta(Z)$ is positive with non-zero probability, there is no guaranty that, after the stratification, $\theta_t >0$ will hold for some $t\in[2T]$. For example, it is possible that after averaging across positively associated tables, the resulting marginal log odds ratios can be zero. In contrast, since under $Q^\ast$ the marginal and conditional log odds ratios coincide within each stratum, the corresponding conditional distribution becomes
\begin{equation*}
	g_{Q^\ast}\left(n_{\mcl{S}}(0,0) \mid n_{\mcl{S}}(0,\cdot), n_{\mcl{S}}(\cdot,0), n_{\mcl{S}}(\cdot,\cdot)\right) =
	\prod_{t=1}^{2T} g_{\theta(z_t^\ast)}\left(n(0,0,S_t) \mid n(0,\cdot,S_t), n(\cdot,0,S_t), n(\cdot,\cdot,S_t)\right),
\end{equation*}
where $\theta(z_t^\ast)> 0$ for $1\leq t\leq T$ and $\theta(z_t^\ast)\geq 0$ for $T+1\leq t\leq 2T$.
Note that $\psi_n$ is a function of $n_{\mcl{S}}(x,y)$ only. Therefore,
\begin{align}
	\mbb{E}_{Q^n}\left[1-\psi_n\mid n_{\mcl{S}}(0,\cdot), n_{\mcl{S}}(\cdot,0), n_{\mcl{S}}(\cdot,\cdot)\right]
	 & = \mbb{E}_{g_Q}[1-\psi_n]
	= \sum_{n_{\mcl{S}}(0,0)}(1-\psi_n)(g_Q - g_{Q^\ast} + g_{Q^\ast})  \notag \\
	\label{supp-eqn:gqTV}
	 & \leq \mbb{E}_{g_{Q^\ast}}[1-\psi_n] + 2\mcl{D}_{TV}(g_Q, g_{Q^\ast}).
\end{align}

\emph{Showing $\mbb{E}_{g_{Q^\ast}}[1-\psi_n] \to 0$:}
For simplicity, let $\mu_{t}(\theta)$ and $\sigma^2_{t}(\theta)$ be the mean and variance of $n(0,0,S_t)$ given $\{n(0,\cdot,S_t), n(\cdot,0,S_t), n(\cdot,\cdot,S_t)\}$ with a log odds ratio $\theta$.
To further simplify, we write $n_t = n(0,0,S_t)$.
Abusing the notation, we write $\theta_t^\ast = \theta(z_t^\ast)$.
\begin{align}
	\mbb{E}_{g_{Q^\ast}}[1-\psi_n]
	 & = 1 - \left(\mbb{P}_{g_{Q^\ast}}[M_n > z_{\alpha/2}] + \mbb{P}_{g_{Q^\ast}}[M_n < z_{1-\alpha/2}]\right) \notag \\
	 & \leq \mbb{P}_{g_{Q^\ast}}[M_n \leq z_{\alpha/2}] \notag                                                         \\
	 & =\mbb{P}_{g_{Q^\ast}}\left(
	\frac{\sum_{t=1}^{2T} n_{t} - \sum_{t=1}^{2T} \mu_{t}(0)}{\sqrt{\sum_{t=1}^{2T} \sigma^2_{t}(0)}}
	\leq z_{\alpha/2}
	\right) \notag                                                                                                     \\
	 & =\mbb{P}_{g_{Q^\ast}}\left(
	\frac{\sum_{t=1}^{2T}n_{t} - \sum_{t=1}^{2T} \mu_{t}(\theta_t^\ast)}{\sqrt{\sum_{t=1}^{2T} \sigma^2_{t}(\theta_t^\ast)}}
	\leq -G_TH_T + z_{\alpha/2}H_T
	\right) \label{supp-eqn:GtHt}
\end{align}
where
\begin{equation*}
	G_T = \frac{\sum_{t=1}^{2T} \mu_{t}(\theta_t^\ast) - \sum_{t=1}^{2T} \mu_{t}(0)}{\sqrt{\sum_{t=1}^{2T} \sigma^2_{t}(\theta_t^\ast)}},\quad
	H_T = \frac{\sqrt{\sum_{t=1}^{2T} \sigma^2_{t}(\theta_t^\ast)}}{\sqrt{\sum_{t=1}^{2T} \sigma^2_{t}(0)}}.
\end{equation*}
As $n\to \infty$, $\sigma^2_{t}(\theta_t^\ast)$ is finite since $n_t$ is bounded above by the condition $T\asymp n$.
Therefore, by Lindberg's condition, the LHS of \eqref{supp-eqn:GtHt} converges in distribution to the standard normal distribution.
On the other hand, the RHS of \eqref{supp-eqn:GtHt} approaches negative infinity.
To see this, note that $\mu_{t}(0)$ and $\sigma_{t}(0)$ are all $O(1)$ since $n_{t}$ is bounded.
According to \citet[Corollary~2.1]{kou1996asymptotics}, if $\theta\geq0$, then $\mu_{t}(0) \leq \mu_{t}(\theta) \leq e^{\theta}\mu_{t}(0)$ and $e^{-\theta}\sigma^2_{t}(0) \leq \sigma^2_{t}(\theta) \leq e^\theta \sigma^2_{t}(0)$.
This shows that $\mu_{t}(\theta_t^\ast)$, $\sigma_{t}(\theta_t^\ast)$ are also of the same asymptotic order as $\mu_{t}(0)$ and $\sigma_{t}(0)$ since $\theta_t^\ast$ is finite $^\forall t\in[2T]$ by the assumption.
Since the numerator of $G_T$ is $O(T)$ whereas the denominator is $O(T^{1/2})$, and $\mu_{t}(\theta_t^\ast) \geq \mu_{t}(0)$ for $t \in [2T]$ with the inequality being strict for $t \in [T]$, we have $G_T \to \infty$.
On the other hand, $H_T \in [e^{-\theta^\ast/2}, e^{\theta^\ast/2}]$ where $\theta^\ast = \max_{t\in[2T]} \theta_t^\ast$.
All combined, \eqref{supp-eqn:GtHt} goes to $0$ as $n\to \infty$.

\emph{Showing $\mcl{D}_{TV}(g_Q, g_{Q^\ast}) \to 0$:}
Similar to the proof of Theorem~\ref{thm:tbl22 T1E}, we deal with the Hellinger distance:
\begin{align*}
	\mcl{D}_{TV}(g_Q, g_{Q^\ast}) \leq \sqrt{2} \mcl{D}_{H}(g_Q, g_{Q^\ast})
	\leq \sqrt{2\sum_{t=1}^{2T} \mcl{D}_H^2(g_{\theta_t}, g_{\theta(z_t^\ast)})} \leq 2\sqrt{T \max_{t\in [2T]} \mcl{D}_H^2(g_{\theta_t}, g_{\theta(z_t^\ast)})}.
\end{align*}
where the second inequality follows from the tensorization property of the Hellinger distance, since $g_Q=\prod_{t=1}^{2T}g_{\theta_t}$ and $g_{Q^\ast}=\prod_{t=1}^{2T}g_{\theta(z_t^\ast)}$.
It is straightforward to see that $g_\theta(x\mid n_{\cdot1}, n_{1\cdot}, n_{\cdot\cdot})$, as a member of a single parameter (the log odds ratio) family, satisfies the assumptions of Lemma \ref{supp-lem:thm1-3}. Therefore,
\begin{equation*}
	\mcl{D}_H^2(g_{\theta_t}, g_{\theta(z_t^\ast)})
	=  J(\theta_t) (\theta_t - \theta(z_t^\ast))^2/8 + o\left((\theta_t - \theta(z_t^\ast))^2\right)
\end{equation*}
where $J(\theta_t)$ is the Fisher information of $g_\theta(x\mid n_{\cdot1}, n_{1\cdot}, n_{\cdot\cdot})$ at $\theta = \theta_t$.
Note that $J(\theta_t)$ is finite as $\theta(z)$ is finite by the conditions in Theorem~\ref{thm:tbl22 Power}.
Furthermore, since the Hellinger smoothness (marginal smoothness) implies condition (ii) of Lemma~\ref{supp-lem:thm1-2}, along with condition (ii) of Theorem~\ref{thm:tbl22 Power}, we can invoke Lemma \ref{supp-lem:thm1-2} to state that $|\theta_t - \theta(z_t^\ast)|=O(h)$.
All combined,
\begin{equation*}
	\sqrt{T \max_{t\in [2T]} \mcl{D}_H^2(g_{\theta_t}, g_{\theta(z_t^\ast)})} = \sqrt{TO(h^2)} \leq \sqrt{n} O(h),
\end{equation*}
so if $h=o(n^{-1/2})$, then $\mcl{D}_{TV}(g_Q, g_{Q^\ast}) \to 0$ as $n\to\infty$.

Alternatively, suppose $Z$ meets the condition (i) of Theorem~\ref{thm:tbl22 Power}.
Letting $\mcl{S}$, $Q^\ast$ and $z_t^\ast$ as previously defined, we can write
\begin{align*}
	 & \mcl{D}_{TV}(\tilde{Q}^n,\tilde{Q}_0^{n})
	\leq \sqrt{2}\mcl{D}_{H}(\tilde{Q}^n,\tilde{Q}_0^{n})
	\leq \sqrt{2n}\mcl{D}_{H}(\tilde{Q},\tilde{Q}_0) \\
	 & \leq
	\sqrt{n}\left\{
	\sum_{t=1}^{2T} s(t)\sum_{X,Y}
	\left(
	\sqrt{\mbb{E}_{Z_t} Q_{XY\mid Z}(X,Y\mid Z_t)} -
	\sqrt{Q_{XY\mid Z}(X,Y\mid z_t^\ast)}
	\right)^2
	\right\}^{1/2}                                   \\
	 & \stackrel{(a)}{\leq}
	\sqrt{6n}\left\{
	\sum_{t=1}^{2T} s(t) L^2h^2
	\right\}^{1/2}
	= \sqrt{6n}Lh
\end{align*}
where $(a)$ holds since $z_t^\ast$ is defined such that $Q_{XY\mid Z}(0,0\mid z_t^\ast) = \mbb{E}_{Z_t}Q_{XY\mid Z}(0,0\mid Z_t)$, and for $(X,Y) = (0,1), (1,0), (1,1)$, we can choose yet another value in $\opn{cl}(S_t)$ that allows us to express the expectations as function values, say $\mbb{E}_{Z_t}Q_{XY\mid Z}(0, 1\mid Z_t) = Q_{XY\mid Z}(0, 1 \mid z_t^{\ast\ast})$ for $z_t^{\ast\ast}\in S_t$.
\qed

\subsection{Proof of Theorem~\ref{thm:condind}}
It is clear that conditional independence leads to $(k_1,k_2)$-conditional independence. We prove the other direction.
Let a triplet $(X,Y,Z)$ be defined on a probability space $(\Omega, \mcl{B}, \mbb{P})$.
Let $\mcl{B}(\mcl{X})$ and $\mcl{B}(\mcl{Y})$ be Borel $\sigma$-algebra of the supports $\mcl{X}$ and $\mcl{Y}$ of $X$ and $Y$.
For a set $A\in \mcl{B}(\mcl{X})$, we write $[X \in A] = \{w\in \Omega : X(w) \in A\}$ and $\sigma(X) = \{[X\in A]:A\in \mcl{B}(\mcl{X})\}$, likewise for $\sigma(Y)$.

The sequence of dyadic partitions $I_0, I_1, \cdots$ of $\mcl{X}$ and $J_0, J_1, \cdots$ of $\mcl{Y}$, defined as in Definition~\ref{dfn:binary partition}, respectively induces sub-classes of $\mcl{B}$ as
\begin{equation*}
	\mcl{C}_k^\mcl{X} = \{[X\in I]:I\in \mcl{I}_k \},\quad
	\mcl{C}^\mcl{X} = \{[X\in I]:I\in \mcl{I} \}
\end{equation*}
which are $\pi$-systems (if augmented with $\varnothing$), likewise for $\mcl{C}_k^\mcl{Y}$ and $\mcl{C}^\mcl{Y}$.
With these notations, $P_{XY\mid Z}(I, J\mid Z) \stackrel{a.s.}{=}P_{X\mid Z}(I\mid Z) P_{Y\mid Z}(J\mid Z)$ being true for all $k_1, k_2 \geq 0$ is equivalent to
\begin{equation}\label{supp-eqn:IJ-pf1}
	\mbb{P}\{C_X \cap C_Y \mid Z\} \stackrel{a.s.}{=}
	\mbb{P}\{C_X \mid Z\} \mbb{P}\{C_Y \mid Z\}
\end{equation}
for all $C_X \in \mcl{C}^\mcl{X}$ and $C_Y \in \mcl{C}^\mcl{Y}$,
where $\mbb{P}\{C\mid Z\}$ is the conditional probability of $C \in \mcl{B}$ given $\sigma(Z)$.
Our goal is to prove that \eqref{supp-eqn:IJ-pf1} holds for any $C_X\in \sigma(X)$ and $C_Y\in \sigma(Y)$, which is the definition of $X\perp Y \mid Z$.

Fix $C^\ast\in \mcl{C}^\mcl{Y}$ and define
\begin{equation*}
	\mcl{L}^{C^\ast} = \{C_X\in \sigma(X):\mbb{P}(C_X\cap C^\ast \mid Z) \stackrel{a.s.}{=}
	\mbb{P}(C_X\mid Z) \mbb{P}(C^\ast\mid Z)\}.
\end{equation*}
Note that $\mcl{C}^\mcl{X}\subseteq \mcl{L}^{C^\ast} \subseteq \sigma(X)$ where $\mcl{C}^\mcl{X}$ is a $\pi$-system. We claim that $\mcl{L}^{C^\ast}$ is a $\lambda$-system:
\begin{enumerate}[(1)]
	\item $\Omega \in \mcl{L}^{C^\ast}$,
	\item if $C_X \in \mcl{L}^{C^\ast}$, then
	      \begin{align*}
		      \mbb{P}\{(\Omega \backslash C_X)\cap C^\ast\mid Z\} & =
		      \mbb{P}\{C^\ast \mid Z\}-\mbb{P}\{C_X \cap C^\ast \mid Z\}              \\&=
		      \mbb{P}\{C^\ast \mid Z\}-\mbb{P}\{C_X \mid Z\} \mbb{P}\{C^\ast \mid Z\} \\&=
		      \mbb{P}\{\Omega \backslash C_X \mid Z\} \mbb{P}\{C^\ast \mid Z\},
	      \end{align*}
	\item similarly, we can check if $C_1,C_2,\cdots \in \mcl{L}^{C^\ast}$ are all disjoint, then $\cup_{i\geq 0} C_i \in \mcl{L}^{C^\ast}$.
\end{enumerate}
Hence, by Dynkin's $\pi$-$\lambda$ theorem, $\sigma(\mcl{C}^\mcl{X}) \subseteq \mcl{L}^{C^\ast}$.
Furthermore, since $\sigma(\mcl{I})=\mcl{B}(\mcl{X})$ by construction,
\begin{equation*}
	\sigma(\mcl{C}^\mcl{X}) =
	\sigma(X^{-1}(\mcl{I})) = X^{-1}(\sigma(\mcl{I})) = \sigma(X),
\end{equation*}
and by the same reasoning $\sigma(\mcl{C}^\mcl{Y})=\sigma(Y)$.
All combined, $\sigma(\mcl{C}^\mcl{X}) =\sigma(X) \subseteq \mcl{L}^{C^\ast}$, hence $\mcl{L}^{C^\ast}=\sigma(X)$. Since $C^\ast\in \mcl{C}^\mcl{Y}$ is arbitrary, \eqref{supp-eqn:IJ-pf1} holds for any $C_X\in \sigma(X)$ and $C_Y\in \mcl{C}^\mcl{Y}$. Repeating the same argument after fixing $C^\ast \in \sigma(X)$, we can see that \eqref{supp-eqn:IJ-pf1} holds for any $C_X\in \sigma(X)$ and $C_Y\in \sigma(Y)$.
\qed

\subsection{Proof of Theorem~\ref{thm:smalltobig}}
It is clear that $X\perp_{k_1,k_2}Y \mid Z$ implies $\theta(I, J, Z)\stackrel{a.s.}{=}0$ for any $I\in \mcl{I}_{k_1-1}$ and $J \in \mcl{J}_{k_2-1}$.
For the other direction, we prove by induction: suppose that the following holds:
\begin{align}\label{supp-eqn:indhyp}
\text{$\theta(I, J, Z)\stackrel{a.s.}{=}0$, $^\forall I\in \mcl{I}_{l_1-1}$, $^\forall J \in \mcl{J}_{l_2-1}$} \quad \Rightarrow \quad \text{$X\perp_{l_1,l_2}Y \mid Z$}.
\end{align}
which is trivially true for $l_1=l_2=1$.
Now, assume that 
\begin{align}\label{supp-eqn:indhyp2}
    \theta(I, J, Z)\stackrel{a.s.}{=}0,\quad ^\forall I\in \mcl{I}_{l_1},\quad ^\forall J \in \mcl{J}_{l_2-1}.
\end{align}
To prove $X\perp_{l_1+1,l_2}Y \mid Z$, pick any $I \in I_{l_1}$ and $J\in J_{l_2-1}$. 
The $2\times2$ table formed by the sets $\{I^{left}, I^{right}\}$ and $\{J^{left}, J^{right}\}$ is independent.
(Note that $I^{left}, I^{right}\in I_{l_1+1}$ and $J^{left}, J^{right}\in J_{l_2}$.)
To see this, for simplicity, we let $p_{ll}$ represent $P_{XY\mid Z}(I^{left}, J^{left} \mid Z)$, and likewise for $p_{lr}, p_{rl}, p_{rr}$.
By the assumption \eqref{supp-eqn:indhyp2}, $\theta(I, J, Z)\stackrel{a.s.}{=}0$ for $I\in I_{l_1}$ and $J\in J_{l_2-1}$, so we have $p_{ll}p_{rr} = p_{lr}p_{rl}$.
If either $p_{rl} + p_{ll}$ or $p_{rr}+p_{lr}$ is $0$, then the table is trivially independent.
Otherwise, a simple manipulation yields
\begin{align*}
    p_{ll}p_{rr} + p_{ll}p_{lr}          & = p_{lr}p_{rl} + p_{ll}p_{lr} \\
	\Rightarrow p_{ll}(p_{rr}+ p_{lr})   & = p_{lr}(p_{rl} + p_{ll})     \\
	\Rightarrow p_{ll}/(p_{rl} + p_{ll}) & = p_{lr}/(p_{rr}+ p_{lr}).
\end{align*}
Hence, the table is independent, and we have $P_{XY\mid Z}(I, J \mid Z)= P_{X\mid Z}(I\mid Z)P_{Y\mid Z}(J \mid Z)$ for $I\in I_{l_1+1}$ and $J\in J_{l_2}$.

Again, by the assumption \eqref{supp-eqn:indhyp2}, $\theta(I, J, Z)\stackrel{a.s.}{=}0$ for $I\in I_{l_1}$ and $J\in J_{l_2-2}$.
We can apply the same logic to conclude $P_{XY\mid Z}(I, J \mid Z)= P_{X\mid Z}(I\mid Z)P_{Y\mid Z}(J \mid Z)$ for $I\in I_{l_1+1}$ and $J\in J_{l_2-1}$.
Repeating this, and combined with the inductive hypothesis \eqref{supp-eqn:indhyp}, we have the result $X\perp_{l_1+1,l_2}Y \mid Z$.
\qed

\subsection{Proof of Theorem~\ref{thm:IJKfacto}}
Conditioned on a set of stratum-specific margin totals $n(I_{k_1}, J_{0}, \mcl{S})$ $n(I_{0}, J_{k_2}, \mcl{S})$, it is clear that
\begin{equation*}
	\mbb{P}_{\tilde{P}_0^n}\left\{ n(I_{k_1}, J_{k_2}, \mcl{S})
	\mid n(I_{k_1}, J_{0}, \mcl{S}), n(I_{0}, J_{k_2}, \mcl{S}) \right\}
	=
	\prod_{S_t\in \mcl{S}}
	\mbb{P}_{\tilde{P}_0^n}
	\left\{
	n(I_{k_1}, J_{k_2}, S_t) \mid
	n(I_{k_1}, J_{0}, S_t), n(I_{0}, J_{k_2}, S_t)
	\right\}.
\end{equation*}
Fix any $S_t \in \mcl{S}$.
For brevity, we suppress explicit notation for dependence on $S_t$ unless needed.
We further introduce a few shortened notations as follows: $n_{ij} = n(I_i, J_j)$ is a $2^{i}\times 2^{j}$ contingency table at strata $S_t$, and $\sigma_{ij}$ is the sigma algebra generated by windows $I\times J$ where $I\in \mcl{I}_i$ and $J \in \mcl{J}_j$.
Lastly, we simplify $\mbb{P}_{\tilde{P}_0^n}$ to $p_n$.
It suffices to show that
\begin{align}\label{supp-eqn:IJfacto}
	\begin{split}
		 & p_n
		\left\{
		n_{k_1, k_2} \mid
		n_{k_1, 0}, n_{0, k_2}
		\right\} \\
		 & =
		\prod_{\substack{I \in \mcl{I}_{k_1-1}, J \in \mcl{J}_{k_2-1}}}
		g_0\left(
		n(I^{left},J^{left}) \mid
		n(I^{left}, J), n(I, J^{left}), n(I,J) \right).
	\end{split}
\end{align}
Although \eqref{supp-eqn:IJfacto} has already been proved in \citet{ma2019fisher}, we provide a slightly more detailed proof for ease of reading.

It is clear that \eqref{supp-eqn:IJfacto} holds for $k_1=k_2=1$.
We will show that if \eqref{supp-eqn:IJfacto} is true for $k_1=i-1$ and $k_2=j$, then it also holds for $k_1=i$ and $k_2=j$.
First, note that
\begin{align*}
	p_n\left\{n_{i, j} \mid n_{i, 0}, n_{0, j} \right\}
	= p_n\left\{n_{i, j} \mid n_{i-1, j}, n_{i, 0}, n_{0, j} \right\}
	p_n\left\{n_{i-1, j} \mid n_{i, 0}, n_{0, j} \right\}.
\end{align*}
For the first term, $p_n\left\{n_{i, j} \mid n_{i-1, j}, n_{i, 0}, n_{0, j} \right\} = p_n\left\{n_{i, j} \mid n_{i-1, j}, n_{i, 0} \right\}$ since $\sigma_{0,j} \subseteq \sigma_{i-1,j}$.
For the second term, since under $\tilde{P}_0$ two margins of the table are independent, $p_n\{n_{a,b} \mid n_{a',0}, n_{0,b}\} = p_n\{n_{a,b} \mid n_{a,0}, n_{0,b}\}$ for any $a'\geq a$.
To see this, we invoke the urn argument in \citet{ma2019fisher}: suppose there are $2^a$ different colors of balls in an urn.
$p_n\left\{n_{a, b} \mid n_{a, 0}, n_{0, b} \right\}$ can be considered the probability of randomly assigning each ball to $2^b$ labels, given the total number of each label.
This can be done by arbitrarily drawing a required number of balls from the urn without replacement for each label.
Assume that, in fact, each color consists of $2^{a'-a}$ different sizes, so that there are $2^{a'}$ unique color and size combinations.
However, since the label was assigned independent of both the color and the size of the balls, conditioning on this new information does not change the probability of label assignment.

All combined,
\begin{align*}
	p_n\left\{n_{i, j} \mid n_{i, 0}, n_{0, j} \right\}
	= p_n\left\{n_{i, j} \mid n_{i-1, j}, n_{i, 0}\right\}
	p_n\left\{n_{i-1, j} \mid n_{i-1, 0}, n_{0, j} \right\}.
\end{align*}
\eqref{supp-eqn:IJfacto} holds for the second term by the inductive hypothesis.
For the first term,
\begin{align*}
	p_n\left\{n_{i, j} \mid n_{i-1, j}, n_{i, 0} \right\}
	 & = p_n\left\{n_{i, j}, n_{i, j-1}, n_{i, j-2}, \cdots, n_{i,1} \mid n_{i-1, j}, n_{i, 0} \right\} \\
	 & = \prod_{b=1}^{j}p_n\left\{n_{i, b} \mid n_{i, b-1}, n_{i-1, b} \right\}                         \\
	 & = \prod_{b=1}^{j} \prod_{I \in I_{i-1}, J \in J_{b-1}}
	p_n\left\{n(I^{left}, J^{left}) \mid n(I^{left}, J), n(I, J^{left}), n(I, J)\right\}
\end{align*}
where the last inequality follows because each window $I\times J$ does not overlap with any other, and the row sums and column sums are given.
Putting all pieces together, we see that
\begin{align*}
	p_n\left\{n_{i, j} \mid n_{i, 0}, n_{0, j} \right\} =
	\prod_{a=1}^{i}\prod_{b=1}^{j} \prod_{I \in I_{a-1}, J \in J_{b-1}}
	p_n\left\{n(I^{left}, J^{left}) \mid n(I^{left}, J), n(I, J^{left}), n(I, J)\right\}
\end{align*}
which is equivalent to \eqref{supp-eqn:IJfacto}.
By the same reasoning, it is easy to check the case when \eqref{supp-eqn:IJfacto} is assumed for $k_1=i$ and $k_2=j-1$.
Therefore, by induction, we have proved \eqref{supp-eqn:IJfacto} for any $i, j \geq 1$.
\qed

\subsection{Proof of Corollary~\ref{cor:pvalindep}}
Let $k = k_1+k_2-2$ and $E(l_1,l_2) = \cap_{I \in I_{l_1}, J \in J_{l_2}}\{p(I, J) \leq \alpha(I, J)\}$ be the events that all the $p$-values computed from $2^{l_1+l_2}$ number of $2\times 2\times T$ tables constituting $2^{l_1+1} \times 2^{l_2+1} \times T$ table (one for each window $I\times J$), written as $n(I_{l_1+1}, J_{l_2+1}, \mcl{S})$, are less than or equal to their corresponding significance levels.
Note that under $\tilde{P}_0$, the conditional probability of the event $E(l_1,l_2)$ is determined by the conditional distribution of $n(I_{l_1+1}, J_{l_2+1}, \mcl{S})$ given $n(I_{l_1+1}, J_{l_2}, \mcl{S})$ and $n(I_{l_1}, J_{l_2+1}, \mcl{S})$.
\begin{align*}
	 & \mbb{E}_{\tilde{P}_0^n}\left[
		\prod_{\substack{l_1+l_2 \leq k}} 1_{E(l_1,l_2)}
		\mid n(I_{k_1}, J_0, \mcl{S}), n(I_0, J_{k_2}, \mcl{S})
	\right]                                            \\
	 & = \mbb{E}_{\tilde{P}_0^n}\left[
		\prod_{\substack{l_1+l_2 < k}} 1_{E(l_1,l_2)}
		\mbb{E}_{\tilde{P}_0^n} \left[
			\prod_{\substack{l_1+l_2 = k}} 1_{E(l_1,l_2)} \mid
			\{n(I_{i}, J_{j}, \mcl{S}):i+j = k+1\}
			\right]
		\mid n(I_{k_1}, J_0, \mcl{S}), n(I_0, J_{k_2}, \mcl{S})
	\right]                                            \\
	 & \stackrel{(a)}{=} \mbb{E}_{\tilde{P}_0^n}\left[
		\prod_{\substack{l_1+l_2 < k}}\!\! 1_{E(l_1,l_2)}\!\!
	\prod_{\substack{I \in I_{l_1}, J \in J_{l_2}      \\l_1+l_2 = k}}\!\! \mbb{E}_{\tilde{P}_0^n} \left[
			1_{p(I, J) \leq \alpha_n(I, J)} \mid
			n(I^{left}, J), n(I, J^{left}), n(I, J)
			\right]
		\mid n(I_{k_1}, J_0, \mcl{S}), n(I_0, J_{k_2}, \mcl{S})
	\right]                                            \\
	 & \stackrel{(b)}{=} \mbb{E}_{\tilde{P}_0^n}\left[
	\prod_{\substack{I \in I_{l_1}, J \in J_{l_2}      \\l_1+l_2 \leq k}} \mbb{P}_{\tilde{P}_0^n} \left[
			p(I, J) \leq \alpha_n(I, J) \mid
			n(I^{left}, J), n(I, J^{left}), n(I, J)
			\right]
		\mid n(I_{k_1}, J_0, \mcl{S}), n(I_0, J_{k_2}, \mcl{S})
	\right]                                            \\
	 & \stackrel{(c)}{=}
	\prod_{\substack{I \in I_{l_1}, J \in J_{l_2}      \\l_1+l_2 \leq k}} \mbb{P}_{\tilde{P}_0^n} \left[
		p(I, J) \leq \alpha_n(I, J) \mid
		n(I^{left}, J), n(I, J^{left}), n(I, J)
		\right]
\end{align*}
where $(a)$ is because the probability of $2^{l_1+1} \times 2^{l_2+1} \times T$ tables given $2^{l_1+1} \times 2^{l_2} \times T$ and $2^{l_1} \times 2^{l_2+1} \times T$ is a product of independent $2\times 2\times T$ tables. $(b)$ follows from recursively applying the same reasoning, noting that conditioning on strata-specific margin totals $n(I_{k_1}, J_0, \mcl{S}), n(I_0, J_{k_2}, \mcl{S})$, the tables $n(I_{l_1}, J_{l_2}, \mcl{S})$ follow the DAG structure shown in Theorem~\ref{thm:IJKfacto}.
$(c)$ is also due to Theorem~\ref{thm:IJKfacto}.
\qed

\subsection{Proof of Theorem~\ref{thm:mainT1E}}
For any window $I\times J$, its $p$-value $p(I, J, \mcl{S}_{IJ})$ is also a function of $n(\{I^{left}, I\}, \{J^{left}, J\}, \mcl{S})$ since $\mcl{S} \preceq \mcl{S}_{IJ}$.
Therefore, $p(I, J, \mcl{S}_{IJ})$ are mutually independent, conditioned on stratum-specific margin totals, as in Corollary~\ref{cor:pvalindep}.
Furthermore, the asymptotic behavior of $p(I, J, \mcl{S}_{IJ})$ and $p(I, J, \mcl{S})$, which will be used in the very last line of this proof, is the same since $T\asymp T_{IJ}$.
For this reason, without loss of generality, we prove the claim under the common stratification. 

The T1E can be decomposed as
\begin{equation*}
	\mbb{E}_{\tilde{P}^n}1_{\tilde{p}\leq \alpha} =
	\int 1_{\tilde{p}\leq \alpha} (\tilde{p}^n_0 - \tilde{p}^n_0 + \tilde{p}^n) d\mu \leq
	\mbb{E}_{\tilde{P}_0^n}1_{\tilde{p}\leq \alpha} + 2 \mcl{D}_{TV}(\tilde{P}^n, \tilde{P}^n_0).
\end{equation*}
By the same reasoning as in the proof of Theorem \ref{thm:tbl22 T1E}, the second term vanishes to $0$ in the limit, as long as $h=o(n^{-1/4})$.
For the first term, we can unravel Šidák's corrections to show that it is equal to $\alpha$ in the limit.
Recall the definition of $\tilde{p}_k$ and $\tilde{p}_{l_1,l_2}$ in Algorithm~\ref{alg:main}.
To ease the notation, we write $k_1' = k_1-1$, $k_2'=k_2-1$, and let $E(k)$, $E(a,b)$ represent the following events:
\begin{align*}
	E(k)   & = \{\tilde{p}_k > 1- (1-\alpha)^{1/(k_1'+k_2'+1)}\},                 \\
	E(a,b) & = \{\tilde{p}_k > 1- (1-\alpha)^{1/(k_1'+k_2'+1)}, a\leq k \leq b\}.
\end{align*}
Finally, we simplify $n(I_{l_1}, J_{l_2}, \mcl{S})$ as $\bm n_{l_1, l_2}$ and omit the subscript $\tilde{P}_0^n$.
\begin{align*}
	\mbb{P}\left\{
	\tilde{p}\leq \alpha \mid \bm n_{k_1, 0}, \bm n_{0, k_2}
	\right\} & = \mbb{P}\left\{
	\min_{0\leq k \leq k_1'+k_2'} \tilde{p}_{k}
	\leq 1 - (1-\alpha)^{1/(k_1'+k_2'+1)}
	\mid \bm n_{k_1, 0}, \bm n_{0, k_2}
	\right\}                        \\
	         & = 1 - \mbb{P}\left\{
	E(0, k_1'+k_2') \mid \bm n_{k_1, 0}, \bm n_{0, k_2}
	\right\}.
\end{align*}
Then
\begin{align*}
	 & \mbb{P}\left\{
	E(0, k_1'+k_2') \mid \bm n_{k_1, 0}, \bm n_{0, k_2}
	\right\}                                \\
	 & = \mbb{E}\left[
		1_{E(0, k_1'+k_2'-1)} \mbb{E} \left[
			1_{E(k_1'+k_2')} \mid \{n(I_i,J_j,\mcl{S}):i+j=k_1'+k_2'+1\}
			\right] \mid \bm n_{k_1, 0}, \bm n_{0, k_2}
	\right]                                 \\
	 & \stackrel{(a)}{=} \mbb{E}\left[
	\prod_{\substack{0\leq k \leq k_1'+k_2' \\ U(k)>0}} \mbb{E} \left[
			1_{E(k)} \mid \{n(I_i,J_j,\mcl{S}):i+j=k+1\}
			\right] \mid \bm n_{k_1, 0}, \bm n_{0, k_2}
		\right]
\end{align*}
where $(a)$ follows from Theorem~\ref{thm:IJKfacto}.
$\mbb{P} \left\{E(k) \mid \{n(I_i,J_j,\mcl{S}):i+j=k+1\}\right\}$ becomes
\begin{align*}
	 & \mbb{P} \left\{
	\tilde{p}_k > 1- (1-\alpha)^{1/(k_1'+k_2'+1)} \mid \{n(I_i,J_j,\mcl{S}):i+j=k+1\}
	\right\}                                    \\
	 & = \mbb{P} \left\{
	\min_{\substack{l_1+l_2=k                   \\ L(l_1, l_2)>0}}
	\tilde{p}_{l_1, l_2} > 1 - (1-\alpha)^{1/(k_1'+k_2'+1) \cdot 1/U(k)}
	\mid \{n(I_i,J_j,\mcl{S}):i+j=k+1\}
	\right\}                                    \\
	 & \stackrel{(b)}{=}
	\prod_{\substack{l_1+l_2=k                  \\ L(l_1, l_2)>0}}
	\mbb{P}\left\{
	\tilde{p}_{l_1, l_2} > 1 - (1-\alpha)^{1/(k_1'+k_2'+1) \cdot 1/U(k)}
	\mid \{n(I_i,J_j,\mcl{S}):i+j=k+1\}
	\right\}                                    \\
	 & \stackrel{(c)}{=}
	\prod_{\substack{l_1+l_2=k                  \\ L(l_1, l_2)>0}}
	\prod_{\substack{I\in I_{l_1}, J\in J_{l_2} \\ V(I, J)=1}}
	\left\{
	1 - F_{IJ, n}\left(1 - (1-\alpha)^{1/(k_1'+k_2'+1) \cdot 1/U(k) \cdot 1/L(i,j)}\right)
	\right\}
\end{align*}
where $(b)$ and $(c)$ are due to Theorem~\ref{thm:IJKfacto} and
\begin{equation*}
	F_{IJ, n}(\alpha) = \mbb{P}\left\{
	p(I, J, \mcl{S}) \leq \alpha
	\mid n(I^{left}, J, \mcl{S}), n(I, J^{left}, \mcl{S}), n(I, J, \mcl{S})
	\right\}.
\end{equation*}
Since two margins of the table $n(I_{k_1}, J_{k_2}, \mcl{S})$ are conditionally independent under $\tilde{P}_0^n$, $\lim_{n\to \infty} F_{IJ, n}(\alpha) = \alpha$ for any $\alpha \in [0,1]$.
All combined,
\begin{align*}
	 & \mbb{P}\left\{
	\tilde{p}\leq \alpha \mid \bm n_{k_1, 0}, \bm n_{0, k_2}
	\right\}                                    \\
	 & =1-  \mbb{P}\left\{
	E(0, k_1'+k_2') \mid \bm n_{k_1, 0}, \bm n_{0, k_2}
	\right\}                                    \\
	 & = 1- \mbb{E}\left[
	\prod_{\substack{0\leq k \leq k_1'+k_2'     \\ U(k)>0}}
	\prod_{\substack{l_1+l_2=k                  \\ L(l_1, l_2)>0}}
	\prod_{\substack{I\in I_{l_1}, J\in J_{l_2} \\ V(I, J)=1}}
		\left\{
		1 - F_{IJ, n}\left(1 - (1-\alpha)^{1/(k_1'+k_2'+1) \cdot 1/U(k) \cdot 1/L(i,j)}\right)
		\right\}
		\mid \bm n_{k_1, 0}, \bm n_{0, k_2}
	\right]                                     \\
	 & \to 1-
	\prod_{\substack{0\leq k \leq k_1'+k_2'     \\ U(k)>0}}
	\prod_{\substack{l_1+l_2=k                  \\ L(l_1, l_2)>0}}
	\prod_{\substack{I\in I_{l_1}, J\in J_{l_2} \\ V(I, J)=1}}
	\left\{
	(1-\alpha)^{1/(k_1'+k_2'+1) \cdot 1/U(k) \cdot 1/L(i,j)}
	\right\}^{(k_1'+k_2'+1) \cdot U(k) \cdot L(i,j)}                                         = \alpha
\end{align*}
as $n \to \infty$ as the integrand is bounded.
\qed

\subsection{Proof of Theorem~\ref{thm:mainPower}}
We prove the local consistency $\lim_{n\to\infty} \mbb{P}_{Q^n}\{ p(I, J, \mcl{S}_{IJ}) \leq \alpha_n(I, J) \mid n(I_{k_1}, J_{0}, \mcl{S}_{IJ}), n(I_{0}, J_{k_2}, \mcl{S}_{IJ}) \} = 1$, which, in turn, implies global consistency $\lim_{n\to \infty} \mbb{P}_{Q^n} \{\tilde{p} \leq \alpha \} = 1$.
Restricting our attention to the window $I\times J$ for some $I\in I_{l_1}$, $J\in J_{l_2}$ satisfying the conditions in Theorem~\ref{thm:mainPower}, we can see that the proof of local consistency is identical to showing how \eqref{supp-eqn:gqTV} vanishes to $0$ in the limit.
Since it is assumed that the induced distribution of a triplet $(1_{X\in I^{left}}, 1_{Y\in J^{left}}, Z)$ satisfies the conditions of Theorem~\ref{thm:tbl22 Power}, and $h=o(n^{-1/2})$, the argument for the RHS of \eqref{supp-eqn:gqTV} remains the same.

The LHS of \eqref{supp-eqn:gqTV} is bounded by \eqref{supp-eqn:GtHt}, where the significance level
\begin{equation*}
	\alpha_n= 1 - (1-\alpha)^{1/(k_1+k_2-1) \cdot 1/U(l_1+l_2) \cdot 1/L(l_1,l_2)}
\end{equation*}
depends on $n$.
If $k_1$ and $k_2$ are some fixed constants, then after some $n$ large enough so that $V(I,J)=1$ for all windows, $\alpha_n(A)$ remains constant as $n\to \infty$; hence, \eqref{supp-eqn:GtHt} decreases to $0$.
If $k_1,k_2$ are of $O(\log n)$, then $\alpha_n = O(1/\log n)$ because $1/b_n\asymp 1-(1-\alpha)^{1/b_n}$ if $\lim_{n\to \infty}b_n= \infty$ and $\alpha\in(0,1)$.
Since $\alpha_n \to 0$, we can write
\begin{equation*}
	z_{\alpha_n/2} = 1-\Phi(\alpha_n/2) \asymp \exp(-\alpha_n^2/8)/\alpha_n\asymp 1/\alpha_n,
\end{equation*}
which shows $z_{\alpha_n/2} = O(\log n)$, where $\Phi(\cdot)$ is the cdf of the standard normal random distribution.
Since $G_T$ in \eqref{supp-eqn:GtHt} is $O(n^{1/2})$, $-G_TH_T+z_{\alpha_n/2}H_T \to -\infty$; hence, \eqref{supp-eqn:GtHt} diminishes to $0$.
\qed

\section{Pseudocode of the stratification algorithm}
The following algorithm produces $2^{\lceil \log_2 \lceil n/\eta\rceil\rceil}$ strata.
To have exactly $T$ strata, we undo the split of the last $2(2^{\lceil \log_2 \lceil n/\eta\rceil\rceil}-T)$ terminal partitions. 
\begin{algorithm}[h!]
	\caption{Recursive dyadic stratification of $\mcl{Z}$}
	\label{alg:medtree}
	\begin{algorithmic}[1]
		\Procedure{medtree}{$\{z_i:i\in [n]\}$, $\eta$} \Comment{$\eta$: desired number of $z_i \in \mbb{R}^d$ per stratum}
		\State $T \leftarrow \lceil n/\eta \rceil$
		\State $\mcl{S}_0 \leftarrow (0,1]^d$
		\For{$t=1, 2, \cdots, \lceil \log_2 T \rceil$}
		\State $j \leftarrow (t-1) \mod d + 1$ \Comment{Cycle through axes 1 to d}
		\State
		Compute $c_j= \opn{median}(\{z_{ij}:z_i \in S\})$ for every $S = \prod_{i=1}^d(a_i,b_i] \in \mcl{S}_{t-1}$ and let
		\begin{align*}
			S^{left} = \prod_{i < j}(a_i,b_i] \times (a_j, c_j] \times \prod_{i > j}(a_i,b_i],\quad
			S^{right} = \prod_{i < j}(a_i,b_i] \times (c_j, b_j] \times \prod_{i > j}(a_i,b_i]
		\end{align*}
		\State $\mcl{S}_{t} \leftarrow \bigcup_{S \in \mcl{S}_{t-1}}\{S^{left}, S^{right}\}$
		\EndFor
		\State \Return $\mcl{S}_{\lceil \log_2 T \rceil}$ \Comment{Return the final stratification}
		\EndProcedure

	\end{algorithmic}
\end{algorithm}
\section{Simulation 3: sensitivity to choice of $\eta$}
\label{supp-sec:sim3}

In this section, we examine how different choices of $\eta$ affect the statistical performance of multiCMH.
Although our method remains consistent against the alternatives satisfying the conditions of Theorem~\ref{thm:tbl22 Power} or Theorem~\ref{thm:mainPower} for all fixed values of $\eta$, its choice influences both T1E and power in finite samples.
Specifically, with small samples, larger strata (higher $\eta$) will better detect conditional dependence under the alternative at the expense of an increased Type I Error (T1E) under the null.
To illustrate these effects, we repeat Simulation~1 in Section~\ref{sec:sim1} for $\eta\in \{5, 10, 15, 20\}$.

As expected, Figure~\ref{supp-fig:sim3 T1E} shows that larger stratum sizes, for a given number of samples, lead to increased T1E. 
This increase does not appear to be affected by the dimension size, as can be seen in the second row of Figure~\ref{supp-fig:sim3 T1E}.
Moreover, larger stratum sizes tend to increase AUROC, particularly for higher-dimensional conditioning sets, as shown in the second row of Figure~\ref{supp-fig:sim3 Power}.
However, as seen in the first row of Figure~\ref{supp-fig:sim3 Power}, for all choices of $\eta$, the AUROC converges to one once the sample size is sufficiently large, despite each stratum size being bounded, thereby confirming our theoretical results.

\begin{figure}[t!]
	\centering
	\includegraphics[width=0.88\linewidth]{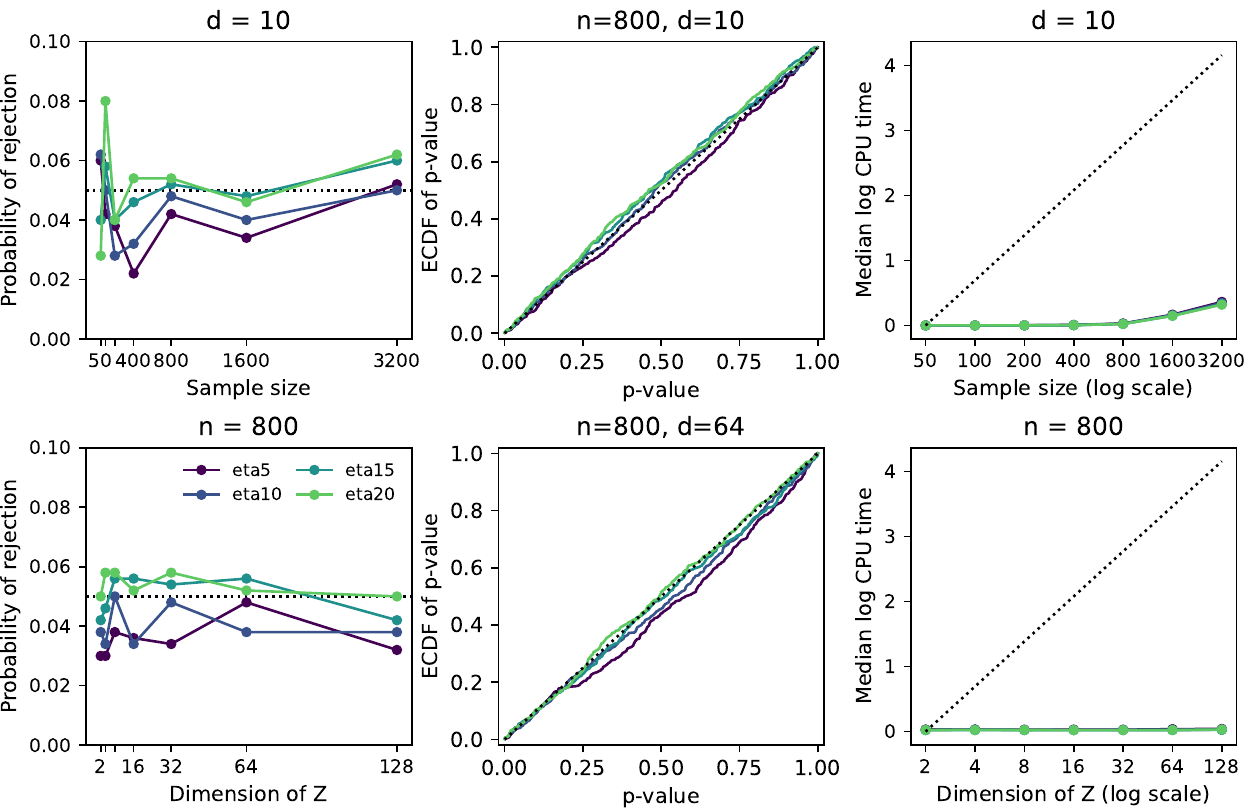}
	\caption{Results of Simulation~3. The nominal level $0.05$ is indicated as a gray dotted line. In the ECDF and median log CPU time plots, a $45^\circ$ reference line through the origin is also drawn in gray dotted line.}
	\label{supp-fig:sim3 T1E}
    \includegraphics[width=0.88\linewidth]{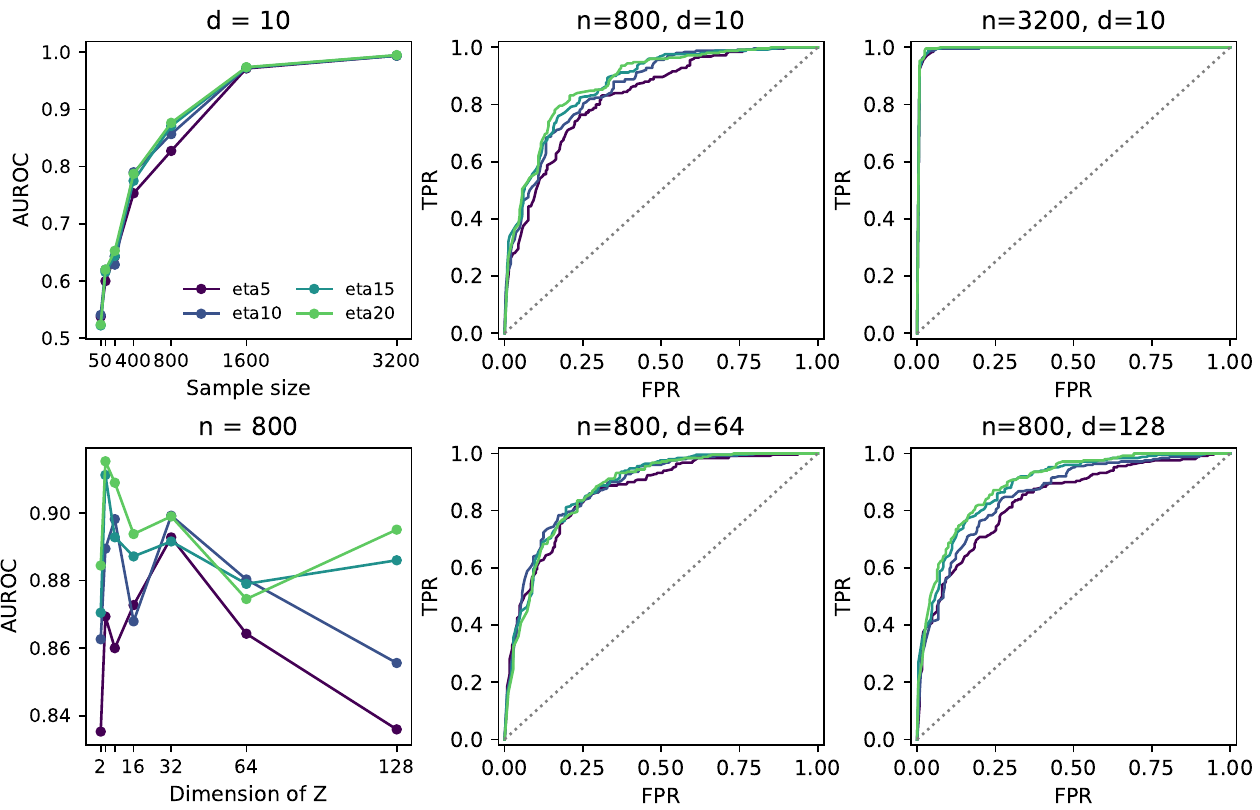}
	\caption{Results of Simulation~3. In the ROC plots, a $45^\circ$ reference line through the origin is included as a gray dotted line.}
	\label{supp-fig:sim3 Power}
\end{figure}

\end{document}